\documentclass[journal,draftcls,onecolumn,12pt,twoside]{IEEEtranTCOM}

\normalsize

\usepackage{cite}

\ifCLASSINFOpdf
   \usepackage[pdftex]{graphicx}
\else
   \usepackage[dvips]{graphicx}
\fi

\usepackage{epstopdf}
   
\usepackage[cmex10]{amsmath}
\usepackage[tight,footnotesize]{subfigure}
\usepackage{url}

\usepackage{amsthm}
\usepackage{amsfonts}

\theoremstyle{plain}
\newtheorem{theorem}{Theorem}
\newtheorem{lemma}{Lemma}
\newtheorem{corollary}{Corollary}

\theoremstyle{definition}
\newtheorem{definition}{Definition}

\theoremstyle{remark}

\theoremstyle{remark}
\newtheorem*{note}{Note}

\theoremstyle{remark}
\newtheorem{remark}{Remark}

\usepackage{mathtools}
\usepackage{cuted}
\usepackage{stfloats}
\usepackage{cleveref}

\usepackage{color}

\newcommand{\dd}{\mathop{}\!\mathrm{d}}

\newcommand{\Rho}{\mathrm{P}}

\makeatletter
\newcommand{\pushright}[1]{\ifmeasuring@#1\else\omit\hfill$\displaystyle#1$\fi\ignorespaces}
\newcommand{\pushleft}[1]{\ifmeasuring@#1\else\omit$\displaystyle#1$\hfill\fi\ignorespaces}
\makeatother

\ifCLASSOPTIONonecolumn
\renewcommand{\baselinestretch}{1.5}  
\newcommand{\figwidth}{0.5\columnwidth}
\newcommand{\figwidthh}{0.45\columnwidth}
\newcommand{\figwidthhh}{0.35\columnwidth}
\newcommand{\figwidthRS}{0.4\columnwidth}
\else
\newcommand{\figwidth}{0.98\columnwidth}  
\newcommand{\figwidthh}{0.85\linewidth}
\newcommand{\figwidthhh}{0.7\linewidth}
\newcommand{\figwidthRS}{\linewidth}
\fi

\begin{document}
%
\title{LDPC Code Design for Distributed Storage: Balancing Repair Bandwidth, Reliability and Storage Overhead}
%
%
%

\author{Hyegyeong~Park,~\IEEEmembership{Student~Member,~IEEE,}
	Dongwon~Lee,
	and~Jaekyun~Moon,~\IEEEmembership{Fellow,~IEEE}
	\thanks{This work has been submitted to the IEEE for possible publication. Copyright may be transferred without notice, after which this version may no longer be accessible. This work was supported by the National Research Foundation of Korea under grant no. NRF-2016R1A2B4011298. This paper was presented in part at the IEEE International Conference on Communications (ICC), 2016. The authors are with the School of Electrical Engineering, Korea Advanced Institute of Science and Technology (KAIST), Daejeon, 34141 South Korea (e-mail: parkh@kaist.ac.kr; leedw1020@kaist.ac.kr; jmoon@kaist.edu).}}
\maketitle

\begin{abstract}
Distributed storage systems suffer from significant repair traffic generated due to frequent storage node failures. 
This paper shows that properly designed low-density parity-check (LDPC) codes can substantially reduce the amount of required block downloads for repair thanks to the sparse nature of their factor graph representation. 
In particular, with a careful construction of the factor graph, 
both low repair-bandwidth and high reliability can be achieved for a given code rate.
First, a formula for the average repair bandwidth of LDPC codes is developed.
This formula is then used to establish that the minimum repair bandwidth can be achieved by forcing a 
regular check node degree in the factor graph.
Moreover, it is shown that given a fixed code rate, the variable node degree should also be regular to yield minimum repair bandwidth,
	under some reasonable minimum variable node degree constraint.
It is also shown that for a 
given repair-bandwidth requirement, 
LDPC codes can yield substantially higher reliability than currently utilized Reed-Solomon (RS) codes. 
Our reliability analysis is based on a formulation of the general equation for the mean-time-to-data-loss (MTTDL) 
associated with LDPC codes. 
The formulation reveals that the stopping number is closely related to the MTTDL. 
It is further shown that LDPC codes can be designed such that a small loss of repair-bandwidth optimality may be traded for a large improvement in erasure-correction capability and thus the MTTDL.

\end{abstract}


%
\IEEEpeerreviewmaketitle

\section{Introduction}
%
%
%
%
\IEEEPARstart{D}{istributed} storage has been deployed as a solution to storing and retrieving massive amounts of data.
By using the MapReduce architecture \cite{MapReduce}, the distributed feature of recent storage systems enables
data centers to store big data sets reliably while allowing scalability and offering high bandwidth efficiency.
However, since distributed storage systems consist of commodity disks, failure events occur frequently. 
As a case in point, in the Google File System (GFS) ``component failures are the norm rather than the exception'' \cite{GFS}.
Simply replicating data multiple times prevent data loss against the node failure events in GFS \cite{GFS} and Hadoop Distributed File System (HDFS) \cite{HDFS}, but the associated 
costs in terms of storage overhead are rather high.

In order to reduce the large storage overhead of replication schemes, erasure codes have been introduced as alternatives \cite{erasure}. 
Reed-Solomon (RS) codes \cite{RS} are typical erasure codes having the maximum distance separable (MDS) property that can tolerate a certain maximum number of erasures given a number of parity blocks. 
Typically, an $(n, k)$ RS code splits a file to be stored into $k$ blocks and encodes them into $n= k+m$ blocks including $m$ parity blocks \cite{Rashmi}. These $n$ blocks of a code are referred to as a {\em stripe} in distributed storage. 
Any $k$ out of $n$ blocks can be used to reconstruct the original file, which is exactly how the MDS property is defined. 
In practice, a (14, 10) RS code is implemented on the Facebook clusters \cite{Rashmi} whereas a (9, 6) RS code is used in the GFS \cite{Azure}. Both of these codes have high storage efficiency as well as orders of magnitude higher reliability compared to 3-replication \cite{erasure, Avail}. Hence, erasure coding schemes based on RS codes have become popular choices especially for archival storage systems where maintaining optimal tradeoff between data reliability and storage overhead is priority.

However, the point at issue is that MDS codes such as RS codes require high bandwidth overhead for the repair process. If a node failure event happens, the erased blocks need to be reconstructed in order to retain the same level of reliability; the amount of blocks to be downloaded for this repair task is defined as {\em repair bandwidth}. 
Since the repair bandwidth represents a limited and expensive resource for data centers, bandwidth overhead associated with the repair job should be carefully managed. For a typical $(n, k)$ RS code, $k$ blocks are required to reconstruct a failed block whereas replication schemes need only one block. 
For instance, the (14, 10) RS code has a 10x repair bandwidth overhead relative to a replication scheme, consuming a significant amount of bandwidth
during repair as confirmed by real measurements in the Facebook's clusters \cite{Rashmi}. 

A number of recent publications have dealt with the repair bandwidth issues.
Dimakis et al. \cite{Dimakis} showed repair models of MDS codes for functional repair and exact repair. Whereas exact repair restores the failed blocks by generating blocks having exact copies of the data, functional repair generates blocks that can be different from the failed blocks as long as the MDS property is maintained. They established optimal storage-bandwidth tradeoff for functional repair and coined the term 
\textit{regenerating codes} for the codes that achieve optimality in this sense. Many researchers have since designed regenerating codes for exact repair that operate in some specific environments \cite{ExactRepair1, ExactRepair2}. 

In contrast to existing works, this paper focuses on coding schemes that offer significant reliability advantages, while achieving highly competitive repair bandwidth and storage overhead tradeoff.
Local reconstruction codes/locally repairable codes (LRCs) and piggybacked RS codes are known methods aiming at reducing repair bandwidth sharing the same key idea. LRCs are non-MDS codes that add local parity symbols to existing RS codes to reduce repair bandwidth at the expense of an increased parity overhead. Windows Azure Storage (WAS) by Microsoft \cite{Azure} and HDFS-Xorbas by Facebook \cite{Xoring, Papailiopoulos14} are practical applications for LRCs\footnote{The Azure-LRC and the Xorbas-LRC are represented by ($k$, $l$, $r_1$) and ($k$, $n - k$, $r_2$), respectively, where $l$ denotes the number of local groups, $r_1$ represents the number of global parities and $r_2$ indicates the block locality.}. Rashmi et al. \cite{Rashmi} suggested piggybacked RS codes which can reduce the repair bandwidth of the RS codes without using extra storage but at the expense of code complexity.

This paper specifically explores design of low-density parity-check (LDPC) codes \cite{LDPC} for distributed storage applications exploiting tradeoffs of the key performance metrics such as repair bandwidth, reliability and storage overhead.
LDPC codes have been considered as an alternative for conventional distributed storage coding schemes. 
However, most known works in this area have been about 
reducing the coding overhead factor of LDPC codes rather than the repair bandwidth \cite{Plank1, Plank2}. 
Whereas Wei et al. \cite{Auto, Wei15, Wei16} showed a low latency of LDPC codes based on simulation and suggested that LDPC codes may have advantages in repair bandwidth and reliability, there has been no rigorous analysis for repair bandwidth and reliability except in \cite{Lee16}.

In \cite{Lee16}, the present authors have demonstrated that LDPC codes provide benefits in terms of both repair bandwidth and reliability given the same storage overhead. Since a variable node of LDPC codes is connected to a relatively small number of nodes, 
LDPC codes have inherent local repair property as LRCs. The repair bandwidth of an LDPC code does not depend on 
the length of the code. The reliability typically gets better with increasing code length. 
Thus, in the case of LDPC codes, the code length can be allowed to grow to achieve excellent reliability 
without worrying about expanding repair bandwidth as in RS codes. The only limiting factor in growing the 
code length in the LDPC codes is the computation and buffer requirements, but compared to the RS codes, the implementation complexity/buffer requirements
of the LDPC codes grow considerably slower with code length. 

This paper refines and adds to the analysis of \cite{Lee16}. Optimality associated with the regularity of the LDPC codes and dependency of LDPC codes' reliability on the stopping set are made precise in the form of theorems with proofs. In addition, this paper also finds LDPC codes that allow a control of the repair bandwidth while targeting high reliability. It is shown that properly designed LDPC codes can achieve very high mean-time-to-data-loss (MTTDL) at the slight expense of the repair bandwidth overhead.

Overall, the key contributions of this paper are as follows.
The average repair bandwidth of the LDPC codes is formulated which leads to the observation 
that a regular check node degree achieves the minimum repair bandwidth given a fixed total number of edges in the factor graph.
Moreover, given a fixed code rate, the variable node degree is also forced to be regular for the repair bandwidth minimality, under some reasonable minimum variable node degree constraint. 
For reliability analysis, a general formula for the MTTDL of LDPC codes is derived. 
The formula shows how the stopping number of a code directly affects reliability. It is confirmed that 
increasing the stopping number of the factor graph greatly enhances reliability. 
Regular quasi-cyclic (QC) progressive-edge-growth (PEG) LDPC codes with different code rates have been designed and compared against representative RS codes and their variants.
The results show that with LDPC codes a slight relaxation of the repair bandwidth minimality 
may allow meaningful improvement in reliability.
In summary, LDPC codes could be a powerful choice for distributed storage systems enjoying both reasonably low repair-bandwidth and very high MTTDL.

The rest of this paper is organized as follows. 
In Section II, we give preliminary information on LDPC codes.
Section III provides repair bandwidth analysis of LDPC codes. 
In Section IV, a design of LDPC codes for high reliability and reasonable repair bandwidth efficiency is discussed. 
In Section V, reliability analysis of LDPC codes is given. Approaches to increase reliability are introduced as well. In Section VI, some specific examples of LDPC codes are discussed which show great performance on distributed storage. Simulation results that compare LDPC codes with other schemes are also given in this section. Finally, the paper draws conclusions in Section VII.

\section{Preliminaries}
\subsection{LDPC Codes}
An LDPC code is a class of linear block codes defined as the null space of an $m \times n$ sparse parity check matrix $\mathbf{H}$ (i.e., $\mathbf{c}\mathbf{H}^T = 0$ if and only if $\mathbf{c}$ is a codeword), where $m$ is the number of parity blocks and $n$ is the length of the codeword. 
The LDPC codes we are concerned with in this paper are binary, which are defined over GF(2).
$\mathbf{H}$ can be illustrated by a factor graph in Fig. \ref{LDPC-wo}, which consists of the check nodes (squares), variable nodes\footnote{Since each variable node stores a coded data block, data node and the variable node are used synonymously in this paper.} (circles), and edges (lines between squares and circles) \cite{Peeling}.
The factor graph describing the LDPC code is called a bipartite graph since it consists of two kinds of nodes: variable node (VN) and check node (CN).
In a bipartite graph of an LDPC code, there are $m$ CNs representing $m$ parity check equations and $n$ VNs indicating $n$ coded blocks. 
CN $r$ is connected to VN $u$ (i.e., CN $r$ involves VN $u$), if $h_{ru}$, the element of $\mathbf{H}$, is 1. 


\begin{figure}[!t]
	\centerline{\includegraphics[width=\figwidthh]{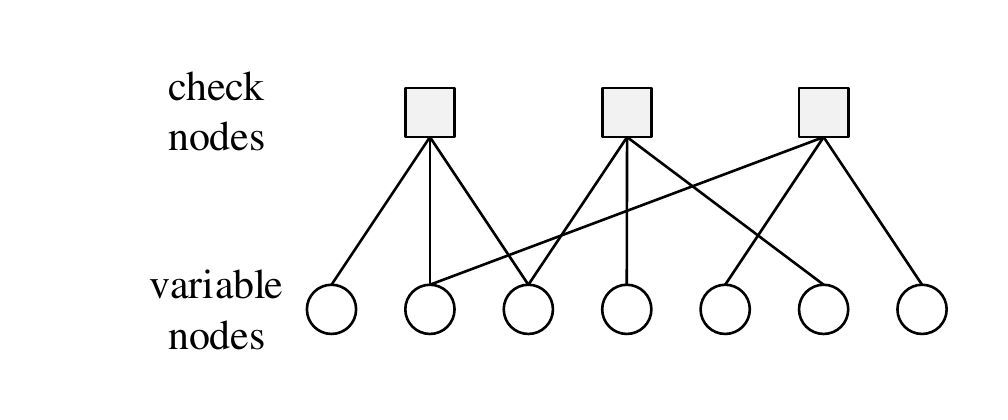}}
	\caption{A graphical representation of an LDPC code
	} \label{LDPC-wo}
\end{figure}

The ensemble of LDPC codes is specified with a variable degree distribution polynomial $\lambda(x)$ and a check degree distribution polynomial $\rho(x)$ \cite{Peeling},  
\begin{equation*}
\label{eq: degree_pair}
\lambda(x) = \sum_{d \geq 2} \lambda_d x^{d - 1} \text{ and } \rho(x) = \sum_{d \geq 2} \rho_d x^{d - 1}\,,
\end{equation*}
where $\lambda_d \text{ (resp. } \rho_d$) is the fraction of edges connected to VNs (resp. CNs) with degree $d$ and $\lambda(1) = \rho(1) = 1$ (i.e., the sum of coefficients is equal to one). This definition of the degree distribution pair $(\lambda, \rho)$ is based on the ``edge perspective".
If the number of edges connected to each VN/CN is all identical, which means that the number of nonzero elements in each row and column in $\mathbf{H}$ are both constant, the corresponding LDPC code is termed a regular LDPC code.
Otherwise, it is designated an irregular LDPC code.

Assuming that the parity check matrix $\mathbf{H}$ is full rank, the code rate $R$ of an LDPC code can be represented by its degree distribution pair $(\lambda, \rho)$ \cite{Richardson01},
\begin{equation}\label{eq: rate_dd}
R = 1 - \frac{m}{n} = 1 - \frac{\int_0^1 \rho(x) \dd x}{\int_0^1 \lambda(x) \dd x} = 1 - \frac{\sum\limits_{d \geq 2} \rho_d/d}{\sum\limits_{d \geq 2} \lambda_d/d}\,.
\end{equation} 

The degree distributions from a node perspective can be expressed from an edge perspective description:
\begin{equation}\label{eq: node_persp}
\Lambda_d = \frac{\lambda_d/d}{\int^1_0 \lambda(x) \dd x} \text{ and } \Rho_d = \frac{\rho_d/d}{\int^1_0 \rho(x) \dd x}\,,
\end{equation}
where $\Lambda_d$ and  $\Rho_d$ are the fractions of VNs and CNs, respectively, with degree $d$.

\subsection{Density Evolution}
Density evolution is a deterministic numerical tool which tracks the fraction of erased variable nodes as iterative decoding proceeds.
For the binary LDPC codes we are concerned with, the failure of a data block can be translated into an erased variable node over binary erasure channel (BEC). 
In this case, the expected fraction of erased data nodes at the $l$-th iteration, $P_l$, as the block length goes to infinity is presented as the recursion \cite{Peeling}:
$$P_l = P_0 \lambda(1 - \rho(1 - P_{l - 1})) \text{ for } l \geq 1\,.$$
Here, for the channel erasure probability of BEC $\epsilon$, $P_0 = \epsilon$.
Decoding with an LDPC code constructed by a degree distribution pair $(\lambda, \rho)$ and an initial erasure probability $\epsilon$ is successful if and only if $\lim\limits_{l \rightarrow \infty} P_l = 0$.
This condition for successful decoding is equivalent to 
$$\epsilon \lambda(1 - \rho(1 - x)) < x \text{ for } x \in (0, \epsilon]\,.$$
 
As the block length grows to infinity, every code in an ensemble tends to behave in the same way. 
Assuming that the code is sufficiently large, the performance of an individual code thus can be captured in the performance of the ensemble average. 
The decoding threshold $\epsilon^\ast$ of an LDPC code is the largest $\epsilon$ value for which the above inequality condition is satisfied. 
We can evaluate the decoding threshold of LDPC codes with the density evolution technique as the block length tends to infinity.  
The decoding threshold divides the channel into areas where data can be reliably stored and areas that are not.
Density evolution therefore provides information on the maximum channel erasure probability that can be corrected by message-passing decoding averaged over all LDPC codes configured by a particular ensemble. This maximum channel erasure probability is called the decoding threshold for the ensemble.

\section{Repair Bandwidth Analysis of LDPC Codes}

In this section, the repair bandwidth of LDPC codes is described in the average sense for all nodes.
LDPC codes are similar to LRCs regarding the repair process since the parity blocks of both codes are made locally from a small portion of data blocks. 
\begin{figure}[!t]
	\centerline{\includegraphics[width=\figwidthhh]{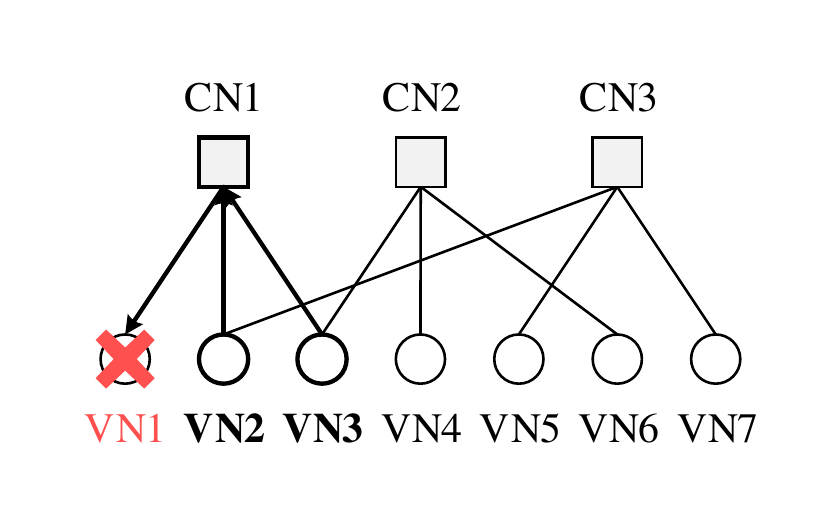}}
	\caption{A block erasure can be represented as a variable node erasure in factor graph. If a block is erased, the erased block can be recovered by downloading other blocks connected to the same check node. VN2 and VN3 are the blocks to be downloaded when VN1 fails.
	} \label{LDPC-fig1}
\end{figure}
As shown in Fig. \ref{LDPC-fig1}, if a block represented by node VN1 is erased, repair job can be done by downloading adjacent blocks VN2 and VN3 connected to the same check node CN1. This simple example demonstrates that LDPC codes can reconstruct erased data by using a relatively small number of blocks.

When an erased block is connected to multiple CNs, as is usually the case, we can choose a specific CN for repair. 
If the LDPC code is regular, any choice is equally good statistically.
For an irregular LDPC code, however, the choice of a CN affects the amount of repair bandwidth since each check may have different degree.
We thus define the bandwidth in the average sense. If a VN is erased, the repair bandwidth for that VN is the number of blocks downloaded averaged over all choices of CNs the VN is connected to. Note that all other VNs are assumed intact in this definition. This value is then averaged over all VN erasure positions. This final average repair bandwidth is obtained by first considering all VNs connected to each CN.
For CN $r$ with degree $d_{c,r}$, there are $d_{c,r}$ VNs attached to it, each of which will have a repair bandwidth of $d_{c,r} - 1$, assuming 
the other VNs attached to CN $r$ are downloaded for repair. The total repair bandwidth associated with CN $r$ can be said to be equal to
$d_{c,r}(d_{c,r}-1)$. Summing over all $m$ CNs, we get $\sum_{r=1}^{m}d_{c,r}(d_{c,r}-1)$. To get to the per-VN repair bandwidth, we recognize that each VN is counted as many times as its node degree in the computation of $\sum_{r=1}^{m}d_{c,r}(d_{c,r}-1)$ since each VN is connected to multiple CNs in general. Thus, this sum should be divided by $nd_{v}$, where $d_{v}$ is the average VN degree, to arrive at the per-VN average repair bandwidth we are looking for. But $nd_{v}=md_{c}=\sum_{r=1}^{m}d_{c,r}$, where $d_{c}$ is the average CN degree. Note that $nd_{v}$ also represents the total number of edges, $E$, in the factor graph. We establish a definition: 

\begin{definition}
	The average repair bandwidth or simply repair bandwidth $\gamma$ of an LDPC code is defined as
	\begin{equation}\label{repair_bandwidth}
	\gamma = \frac{\sum_{r=1}^{m}d_{c,r}(d_{c,r}-1)}{E}\,,
	\end{equation}
	where $m$ represents the number of parity blocks of an LDPC code, $d_{c,r}$ denotes the degree of check node $r$ and $E$ indicates the total number of edges in the factor graph.
\end{definition}

The following lemma subsequently tells us how the CN degrees should be distributed to minimize the average repair bandwidth of
(\ref{repair_bandwidth}).

\begin{lemma}\label{regular_check}
	Given a fixed number $E$ of edges on the factor graph, a regular check node degree minimizes the repair bandwidth of LDPC codes to $d_c-1$, where $d_c$ denotes the check node degree of the corresponding LDPC codes.
\end{lemma}

\begin{IEEEproof}
	The repair bandwidth in \eqref{repair_bandwidth} can be rewritten as 
	$$\gamma = \frac{\sum_{r=1}^{m}(d_{c,r}-\frac{1}{2})^2 - \sum_{r=1}^{m}(\frac{1}{2})^2}{E}\,.$$
	By using the Cauchy-Schwarz inequality and the constraint $\sum_{r = 1}^{m} d_{c, r} = E$, the choice $$d_{c,1}=d_{c,2}=\cdots=d_{c,m}=E/m$$ minimizes the average bandwidth.
	Thus, a regular CN degree minimizes the repair bandwidth and the corresponding minimum value is $\gamma_\text{min}=d_{c}-1$, one less than the CN degree.
\end{IEEEproof}

Lemma \ref{regular_check} indicates that an LDPC code must be CN-regular in order to minimize the repair bandwidth. How about the VN degrees? Before discussing desirable VN degree characteristics in light of the repair bandwidth issue, it is natural to impose a minimum VN degree constraint such that any VN in a factor graph has a degree at least equal to some positive integer $d_{v_\text{min}}$. This is due to practical reasons having to do with decodability. For example, we obviously need $d_{v_\text{min}}=1$ so that each VN is attached to at least one CN, in order to reconstruct any single VN erasure. In practical applications where a node may fail before the current failure can be repaired, we actually need a more stringent condition: LDPC codes with even degree-1 VNs have been deemed impractical \cite{Divsalar09, Nguyen12, Garcia03}, suggesting that we should set $d_{v_\text{min}}=2$. This type of minimum VN degree requirement calls for the VN-regularity as well. We summarize the desired CN and VN degree characteristics in the following combined statement. 

\begin{theorem}\label{regular}
	Among the factor graphs having no VNs with degree less than $d_{v_\text{min}}$, a chosen graph yields 
	an LDPC code of rate $R$ with minimum repair bandwidth if and only if it is both CN- and VN-regular with $d_c=d_{v_\text{min}}/(1-R)$ and $d_v=d_{v_\text{min}}$.
\end{theorem}
\begin{IEEEproof}
Lemma 1 states that a minimum-repair-bandwidth LDPC code is CN-regular with the uniform CN degree of $d_c$.
The proof follows directly from this lemma combined with the fact that the ratio of $d_v$, the average VN degree, to $d_c$ gets fixed once the code rate is given as seen in the relation: $d_v/d_c=1-R$, where $R = k/n= (n-m)/n$. For a given $R$ value, the average VN degree $d_v$ must be made as small as possible to minimize $d_c$ so that 
\begin{equation}\label{regular_LDPC_equation}
\gamma_\text{min} =d_{c}-1= \frac{d_{v}}{1-R} - 1
\end{equation}
is in turn minimized. But since each VN degree is greater than or equal to $d_{v_\text{min}}$ by assumption, so is the average VN degree $d_v$. Apparently, the minimum average VN value of $d_v=d_{v_\text{min}}$ is achieved when all VNs have a fixed degree of $d_{v_\text{min}}$, i.e., when the factor graph is VN-regular. As for the CN degree, we obviously need $d_c=d_{v_\text{min}}/(1-R)$ for minimum repair bandwidth.	 
\end{IEEEproof}

It is clear that the repair bandwidth of an LDPC code does not depend on the code length, but on $d_{c}$. This property makes the LDPC codes a powerful option for distributed storage. Moreover, for a fixed $d_{v}$, (\ref{regular_LDPC_equation}) also reveals an interesting relationship that 
$\gamma_\text{min}$ increases with increasing $R$, which is due to the fact that for a fixed $d_{v}$, increasing $R$ must also mean increasing 
$d_{c}$.

The regularity of minimum repair bandwidth LDPC codes automatically results in a condition on the code rate, as stated in the following corollary.
\begin{corollary}\label{condition}
An LDPC code of rate $R$ allows minimum repair bandwidth only if $d_{v_\text{min}}/(1-R)$ is an integer greater than $d_{v_\text{min}}$.
\end{corollary}

Theorem \ref{regular} indicates that under the practical constraint of $d_{v_\text{min}}=2$,
regular LDPC codes with $d_v = 2$ give the best
repair bandwidth efficiency. At this point, a useful question arises: 
if we are allowed to increase $d_v$ beyond 2, in hopes of improving reliability for certain applications,
how rapidly do we lose repair-bandwidth efficiency?
In other words, we are interested in investigating
the possibility of relaxing the repair bandwidth minimality in an effort to improve reliability.
Specifically, we shall compare the reliability-bandwidth tradeoffs of the  
regular LDPC codes having a fixed $d_v = 2$ with VN-irregular LDPC codes with average VN degree beyond 2.
In the process, we provide a new VN-irregular LDPC code design that  
allows a good erasure correction capability at the slight expense of the efficiency of the repair bandwidth.
In comparing different coding schemes we consider three code rates: 1/2, 2/3 and 3/4.
Theorem \ref{regular} provides the reference point for minimal repair bandwidth.   
 
\begin{remark}[Non-binary LDPC Codes]
	As can be seen in Fig. \ref{fig. LDPC_qary}, suppose we have an $(n, k ) = (5, 3)$ coded system consisting of words (or symbols), each of which has $q$ data bits in it (e.g., non-binary LDPC codes of GF($2^q$) or RS codes of GF($2^q$)).  
	As an example of non-binary LDPC codes consisting of $q$-bit symbols, the repair bandwidth is given by
	$B/q\cdot(d_c - 1)\cdot q = B(d_c - 1)$,
	where $B$ is the block size and assume for simplicity that $B$ is a multiple of $q$.
	Likewise, for RS codes made up of $q$-bit words, $B/q \cdot k\cdot q = Bk$ bits are required for repairing a failed block.
	From the two examples above, the repair bandwidth of codes consisting of multiple bits does not change with the symbol or word size $q$.
	Hence we shall consider only binary LDPC codes in this paper without loss of generality.
	We will stick to the normalized value $d_c - 1$ for the repair bandwidth instead of the more general expression $B(d_c - 1)$ for simplicity. 
\end{remark}
\begin{figure}[!t]
	\centerline{\includegraphics[width=\figwidthh]{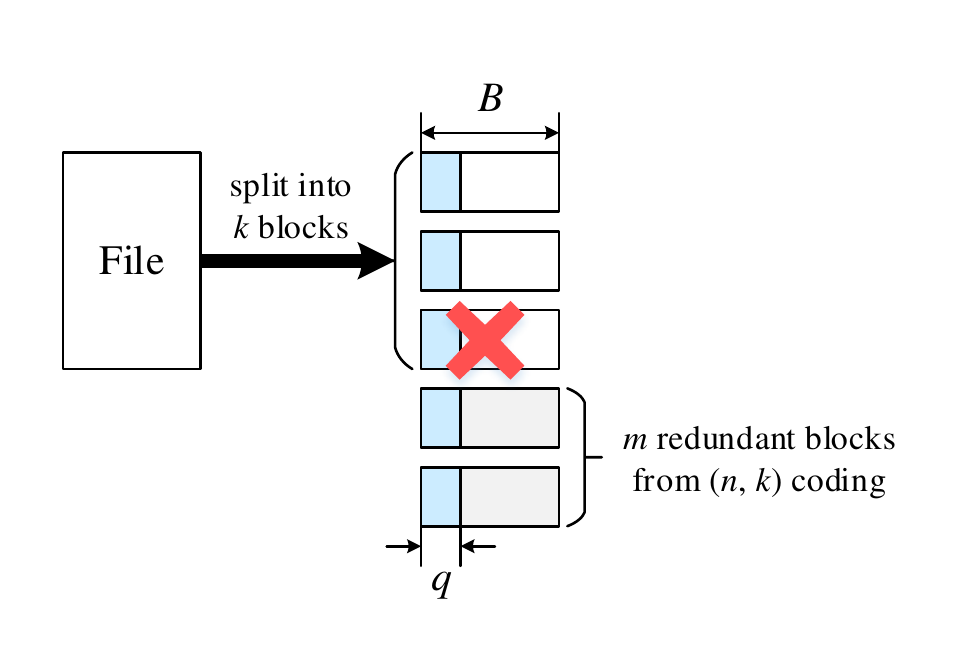}}
	\caption{An example of coded storage using multiple bits per symbol. $B/q$ codes each of which is made up of $q$-bit symbols are employed across blocks of size $B$. This example is in systematic form (i.e., the original data is a part of the coded blocks) for illustrative purpose. 		 
	} \label{fig. LDPC_qary}
\end{figure}

\section{A Design for the Efficiency in Both Repair and Protection} \label{sec: Irregular}
In this section, we suggest a degree distribution design criterion to guarantee high erasure-correction-capability while enjoying reasonable efficiency in repair-bandwidth.  
Recall that a regular CN degree minimizes the repair bandwidth for a given number of edges in the factor graph and the minimum repair bandwidth is given by $\gamma_\text{min} = d_c - 1$ in Lemma \ref{regular_check}. 
While maintaining the CN-regularity, we will relax the regularity condition on VN to find the appropriate VN degree distribution.
The average VN degree $d_v$ can be computed as
\begin{equation}\label{eq: dv}
\begin{split}
d_v& = \sum_d d \Lambda_d \\
&= \sum_d d \frac{\lambda_d/d}{\int^1_0\lambda(x)\dd x} = \frac{1}{\int^1_0\lambda(x) \dd x} = \frac{1}{\sum\limits_d \frac{\lambda_d}{d} }\,,   
\end{split}
\end{equation} 
where $\Lambda_d = \frac{\lambda_d/d}{\int^1_0 \lambda(x) \dd x}$ from \eqref{eq: node_persp} and $\sum\limits_d \lambda_d = 1$.

Using \eqref{eq: dv}, $\gamma_\text{min}$ can be rewritten as
\begin{equation}\label{eq: rbw}\nonumber
\begin{split}
\gamma_\text{min} = d_c - 1 &= \frac{d_v}{1 - R} - 1 \\
&= \frac{1}{1 - R}\cdot \frac{1}{\sum\limits_d \frac{\lambda_d}{d}} - 1 \,.
\end{split}
\end{equation} 

Therefore, we need to maximize $\sum\limits_d \frac{\lambda_d}{d}$ in order to minimize $\gamma_\text{min}$ 
for a fixed code rate constraint of $R$.

We propose a design of LDPC codes that balances the repair bandwidth overhead and the system reliability. To do this, we consider the following optimization problem with the VN degree distribution parameters $\lambda_d$ as optimization variables.
\begin{alignat}{3}
\label{eq: objfunc}& \text{maximize} \quad & \sum_d \frac{\lambda_d}{d} \phantom{aaaaaaaaa} &  \\ 
\label{eq: exit} & \text{subject to} \quad & \epsilon_\text{con} \lambda(1 - \rho(1 - x))& < x,\quad x \in (0, \epsilon_\text{con}] \\
\label{eq: rate} && \sum_d \frac{\lambda_d}{d} & = \frac{1}{1 - R} \sum_d\frac{\rho_d}{d}
\end{alignat}
where $\epsilon_\text{con}$ represents the minimum level of reliability imposed. 
The constraint (\ref{eq: exit}) ensures successful decoding as discussed in Section II.B, and (\ref{eq: rate}) is the rate constraint from \eqref{eq: rate_dd}. In addition, obvious extra constraints exist on any VN degree distribution polynomial:
$\sum_d \lambda_d = 1$ and $0 \leq \lambda_d \leq 1 \,.$

We employ a CN degree distribution $\rho(x) = x^{d_c - 1}$, forcing the CN-regularity.
Hence, \eqref{eq: exit} and \eqref{eq: rate} reduce to the following constraints:
\begin{equation}\label{eq: exit2}
\epsilon_\text{con} \lambda(1 - (1 - x)^{d_c - 1}) < x,\quad x \in (0, \epsilon_\text{con}]\
\end{equation}
\begin{equation}\label{eq: rate2}
\sum_d \frac{\lambda_d}{d} = \frac{1}{1 - R}\cdot \frac{1}{d_c}\,.
\end{equation}
Our optimization problem can then be stated as: for a given value of $\epsilon_\text{con}$, find the distribution 
$\lambda_d$ that will minimize $d_c$ while satisfying \eqref{eq: exit2} and \eqref{eq: rate2}. A small $d_c$ value would be great for maintaining repair efficiency but to tolerate a higher value of $\epsilon_\text{con}$ in ensuring reliability, a compromise would have to be made on how small $d_c$ could get. Noticing that $d_c$ are integer values forming a relatively small search space, a clear picture on this tradeoff can be obtained conveniently by fixing $d_c$ and then iteratively finding the maximum value of $\epsilon_\text{con}$ and the corresponding $\lambda_d$ that satisfy \eqref{eq: exit2} and \eqref{eq: rate2} for each fixed value of $d_c$. Fig. \ref{fig.tradeoff} shows the relationships obtained for the minimum repair bandwidth $d_c-1$ versus the decoding threshold $\epsilon_\text{con}$ for different code rates. Fig. \ref{fig.tradeoff} clearly reveals the maximum level of reliability that can be achieved for a given repair bandwidth or, equivalently, the minimum repair bandwidth attainable for a given level of reliability, for some fixed code rate.

\begin{figure}[!t]
	\centerline{\includegraphics[width=\figwidth]{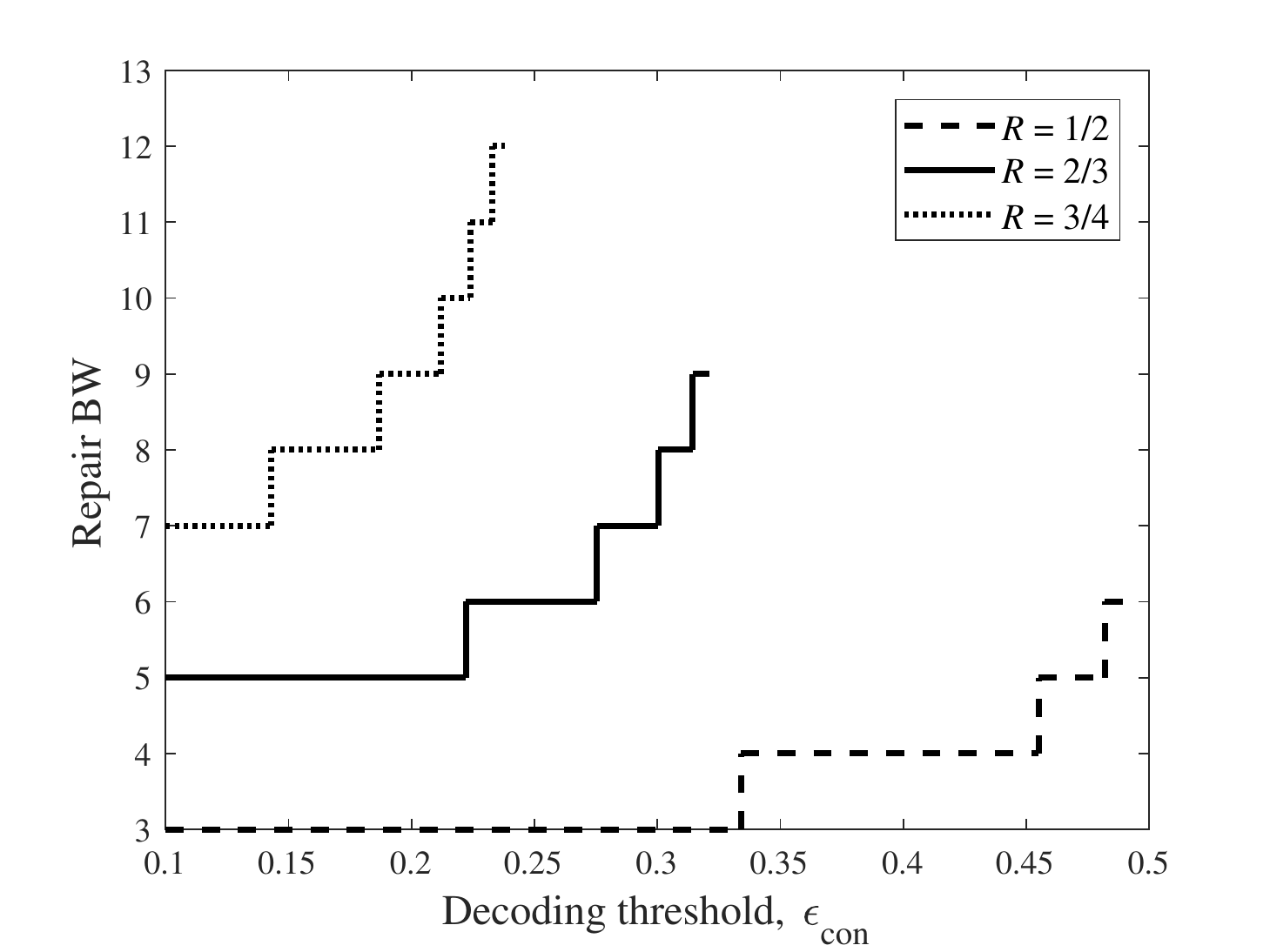}}
	\caption{Repair bandwidth and decoding threshold tradeoff curves for different code rates. 		 
	} \label{fig.tradeoff}
\end{figure}



We also present several examples of designed degree distributions for different target code rates of $R$ = 1/2, 2/3 and 3/4 in Table \ref{degree_distributions}.
The scaled maximum decoding threshold $\epsilon^\ast_\text{con}/(1 - R)$ indicates how close the designed decoding threshold is to the BEC capacity. 
The repair bandwidth $\gamma=d_c - 1$ and the average VN degree $d_v$ are also shown.

\begin{table*}[!t]
	\caption{Examples of designed VN degree distributions
	}
	\begin{center}
		\begin{tabular} {ccccccccccc}
			\hline\hline
			$R$ & $\gamma = d_c - 1$ &  $\epsilon_{\text{con}}^\ast/(1 - R)$ & $d_v$ &$\lambda(x)$\\
			\hline
			1/2 & 3 & 0.6680 & 2 & $x$\\
			1/2 & 4 & 0.9100 & 2.4997 & $0.5496x + 0.1549x^2 + 0.2956x^3$\\
			1/2 & 5 & 0.9640 & 2.9990 & $0.4128x + 0.1789x^2 + 0.1128x^3 + 0.1371x^6 + 0.1584x^7$\\
			1/2 & 6 & 0.9840 & 3.4987 & $0.3394x + 0.1403x^2 + 0.1036x^3 + 0.0940x^5 + 0.0963x^6 + 0.0378x^{14}$\\
			&&&& $ + 0.1886x^{15}$\\
			\hline
			2/3 & 5 & 0.6667 & 2 & $x$\\
			2/3 & 6 & 0.8260 & 2.3328 & $0.5716x + 0.4284x^2$\\
			2/3 & 7 & 0.9010 & 2.6662 & $0.4775x + 0.0880x^2 + 0.4098x^3 + 0.0247x^4$\\
			2/3 & 8 & 0.9430 & 2.9977 & $0.3927x + 0.2279x^2 + 0.2907x^5 + 0.0887x^6$\\
			2/3 & 9 & 0.9640 & 3.3321 & $0.3469x + 0.1440x^2 + 0.1331x^3 + 0.0708x^4 +0.1001x^8 + 0.2051x^9$\\
			\hline
			3/4 & 7 & 0.5720 & 2 & $x$\\
			3/4 & 8 & 0.7480 & 2.2500 & $0.6704x + 0.3296x^2$\\
			3/4 & 9 & 0.8480 & 2.4997 & $0.4548x + 0.4462x^2 + 0.0991x^3$\\
			3/4 & 10 & 0.8960 & 2.7500 & $0.4486x + 0.1325x^2 + 0.2488x^3 + 0.1702x^4$\\
			3/4 & 11 & 0.9320 & 2.9973 & $0.3867x + 0.2270x^2 + 0.3863x^5$\\
			3/4 & 12 & 0.9520 & 3.2468 & $0.3495x + 0.1640x^2 + 0.1432x^3 + 0.3306x^7 + 0.0127x^8$\\
			\hline\hline
		\end{tabular}\label{degree_distributions}
	\end{center}
\end{table*}

\section{Reliability Analysis of LDPC Codes} \label{sec: Reliablity}
\subsection{The Mean Time to Data Loss}
We now provide reliability analysis for LDPC codes.
In particular, we show that increasing the stopping number of the factor graph directly influences reliability.
A Markov model is introduced to estimate system reliability of coding schemes. Continuous-time Markov models have been commonly used to compare reliability of storage systems in terms of the MTTDL (e.g., see \cite{Ramabhadran06, Avail, Azure, Xoring, BinaryLRC16}). Unlike the bit-error-rate (BER) or the word-error-rate (WER)
performance metric, the MTTDL metric based on the Markov model considers the repair speed,
which is our main interest in this paper.

Fig. \ref{Markov_RS} shows a Markov model example of the (14, 10) RS code \cite{Xoring}. 
The MTTDL is mainly influenced by the number of failures which can be tolerated before data loss as well as by the repair rate. 
Here, $\lambda$ indicates the failure rate of a node and $\mu$ represents the repair rate of the nodes. 
Typically, $\mu \gg \lambda$ for storage applications. 
We can assume that each node fails independently at rate $\lambda$ if the blocks are stored in different racks (physically separated storage units in data centers). 
Then, it is reasonable to ignore the possibility of burst failures. Also, the adoption of a continuous-time Markov model presupposes that only a single node failure is allowed at a given instance.
Each state of the Markov model represents the number of erased blocks in a stripe. 
For the (14, 10) RS code, state 5 is the data loss (DL) state since five erasures in a stripe cannot be decoded. 
Whereas the failure rates depend on the state, the repair rates are all the same since the number of blocks to be downloaded for repair is always 10. 
The MTTDL can be obtained from this Markov model by calculating the mean arrival time to the DL state.
The MTTDL of the MDS codes are well-established \cite{Avail, Trivedi}. 
The MTTDL analysis for MDS codes can be modified and extended for the LDPC codes, as discussed next.

\begin{figure}[!t]
	\centerline{\includegraphics[width=\figwidthRS]{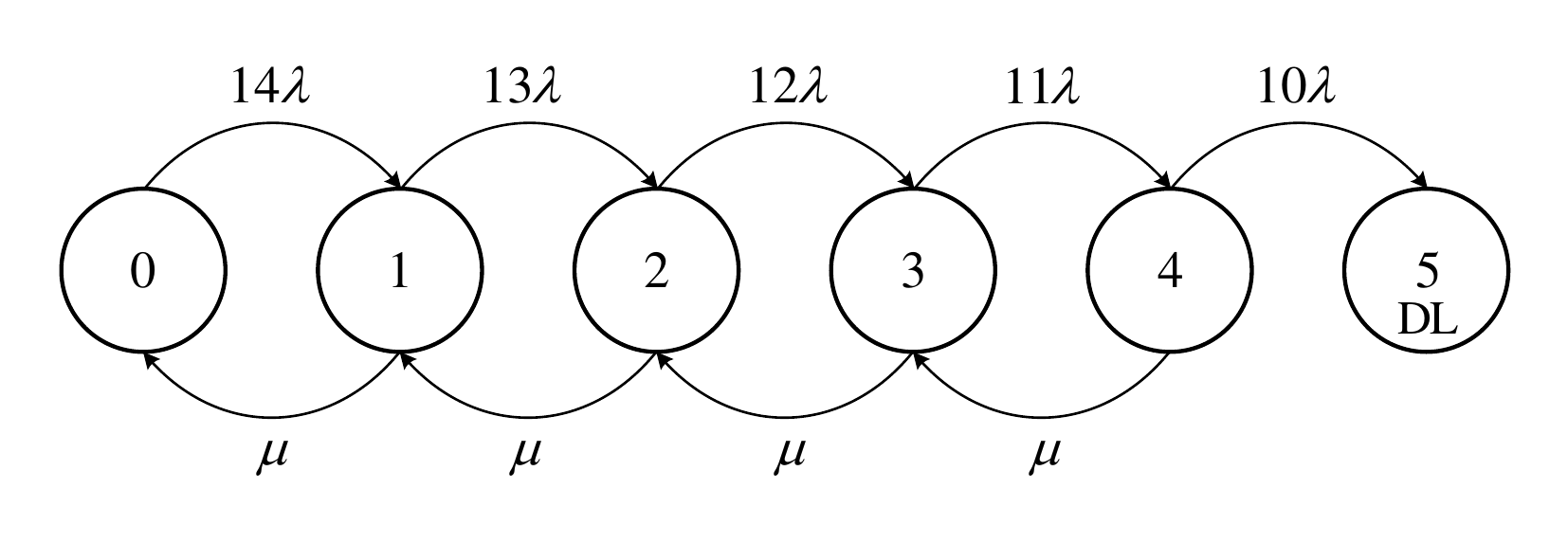}}
	\caption{Markov model of the (14, 10) RS code
	} \label{Markov_RS}
\end{figure}

\subsection{MTTDL of LDPC Codes}
In this section, details in calculating the MTTDL for non-MDS codes are described. 
While the Markov model is already discussed for the LDPC codes in \cite{notes}, the general formula for the MTTDL of the LDPC codes has not been given. 
We provide such a formula here. 
We also 
develop insights into how the MTTDL of the LDPC codes is affected by the stopping number. 
Before presenting the Markov model of LDPC codes, some key terms are clarified. On factor graphs, the girth indicates the shortest cycle. A stopping set \cite{Stopping} is a subset of variable nodes such that all 
check nodes connected to it are connected by at least two edges, and the stopping number is the size of the smallest stopping set.

The derivation process is similar to that for MDS codes. However, as shown in Fig. \ref{non-MDS2}, LDPC codes can directly go to the data loss state with only a small number of erasures. For instance in Fig. \ref{LDPC-fig1}, if VN5 and VN7 fail, it is impossible to repair those nodes unlike in MDS codes. To model this behavior, probability parameters are introduced to the Markov model. Probability  
$p_{i}$ is the conditional probability that a stripe of a given code can tolerate an additional node failure given state $i$. 
This means that the code has already survived from $i$ failures and can tolerate one more failure with probability $p_{i}$.
In general, LDPC codes are designed to guarantee $p_{0}=1$ and $p_{1}=1$ since length-4 cycles are prohibited; however, other probabilities depend on the parity-check matrix of the code. If the parity-check matrix of the LDPC code is given, the $p_{i}$ values can be obtained by the relationship, $p_i = q_{i + 1}/q_i$, where $q_i$ denotes the unconditional probability that a given code can tolerate $i$ failures \cite{notes}. Such unconditional probabilities can be estimated by decoding simulation of LDPC codes on the erasure channel. Exploiting the estimators of $q_i$ and $q_{i + 1}$, we can obtain an asymptotically unbiased estimator of $p_i$ given a large number of samples. This can be justified as follows.

\begin{note}
Given two random variables $Y$ and $Z$, assume that we cannot directly measure the ratio $Y/Z$.
From the measured realizations $y$ and $z$, we wish to estimate $Y/Z$.  
Suppose that $\tilde{y}$ and $\tilde{z}$ are samples means over $n_s$ samples.
Given the $n_\text{s}$ samples of $y$ and $z$, a possible estimator for the ratio $Y/Z$ is
a sample ratio $\tilde{y}/\tilde{z}$.
The bias of this estimator goes as $\mathcal{O}(1/n_\text{s})$: 
\begin{equation} \label{taylor}
\mathbb{E} \left[\frac{Y}{Z}\right] = \frac{\tilde{y}}{\tilde{z}} + \mathcal{O}\left(\frac{1}{n_\text{s}}\right)\,,
\end{equation}
which indicates that $\tilde{y}/\tilde{z}$ is an unbiased estimator as $n_\textrm{s}$ tends to infinity.
\end{note}

\begin{figure*}[!t]
	\centerline{\includegraphics[width=0.65\linewidth]{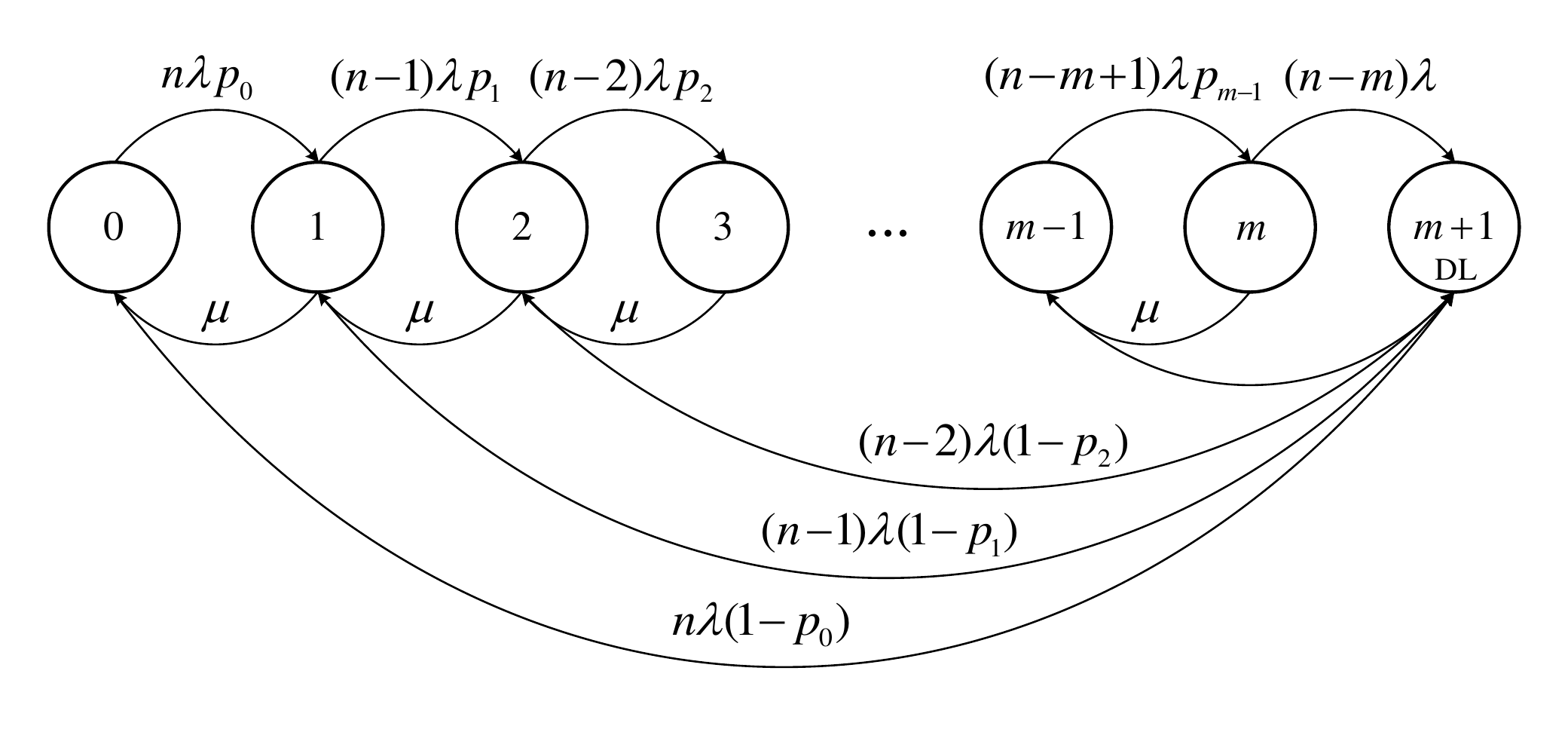}}
	\caption{Markov model of LDPC codes with $m$ parity blocks
	} \label{non-MDS2}
\end{figure*}

For $m$ parity blocks (see Fig. \ref{non-MDS2}), the MTTDL equation is given by the following lemma.

\begin{lemma}\label{MTTDL_parameter}
	For an arbitrary number $m$ of the parity blocks, the MTTDL of $(n, k)$ LDPC codes can be represented by
	\begin{equation}\label{MTTDL_LDPC}
	\text{MTTDL} \rightarrow \frac{(m+1)\mu^{m}}{f(n, m, \lambda, \mu, p_0, \cdots, p_{m-1})} \text{ as } \frac{\lambda}{\mu} \rightarrow 0\,,
	\end{equation}
	where
	$\lambda$ denotes the failure rate of a node, $\mu$ indicates the repair rate of the nodes, $p_i$ represents the conditional probability that a stripe of a given code can tolerate an additional node failure given state $i$ and $f(n, m, \lambda, \mu, p_0, \cdots, p_{m-1})$ is defined by 
	\begin{equation}\label{LDPC_MTTDL_f}
	\begin{split}
	f(n, m, \lambda, \mu, p_0, \cdots, p_{m-1}) &=  n\lambda(1-p_{0})\cdot \mu^{m} \\
	&\hspace{0.3cm} +\sum_{j=1}^{m-1}\bigg[\Big\{\prod_{i=0}^{j}(n-i)\cdot
	\lambda^{j+1}\Big\} \cdot \Big\{\prod_{i=0}^{j-1}p_{i} \cdot (1-p_{j})\cdot \mu^{m-j}\Big\}\bigg] \\
	&\hspace{0.3cm} + \Big\{\prod_{i=0}^{m}(n-i)\cdot\lambda^{m+1}\Big\} \cdot \prod_{i=0}^{m-1}p_{i}\,.
	\end{split}
	\end{equation}
\end{lemma}

\begin{proof}
	See Appendix \ref{proof_MTTDL_parameter}.
\end{proof}

From Lemma \ref{MTTDL_parameter} it is seen that with all other parameters fixed, making the $p_{i}$ values large increases the MTTDL by examining what each term in the MTTDL is doing in the limit. In order to see the behavior in the limit, divide the numerator and denominator of the MTTDL by $\mu^m$ and write: 
\begin{multline}\label{den_limit}
\frac{f(n, m, \lambda, \mu, p_0, \cdots, p_{m-1})}{\mu^m} = \left(\frac{\lambda}{\mu}\right)^0 n\lambda(1 - p_0) + \left(\frac{\lambda}{\mu}\right)^1\lambda(1 - p_1)p_0 \cdot n(n - 1) \\
+ \left(\frac{\lambda}{\mu}\right)^2 \lambda(1 - p_2)p_0p_1 \cdot n(n - 1)(n - 2) \\
+ \cdots + \left(\frac{\lambda}{\mu}\right)^{m - 1} \lambda(1 - p_{m - 1})p_0p_1\cdots p_{m - 2} \cdot n(n - 1)\cdots (n - m + 1) \\
+ \left(\frac{\lambda}{\mu}\right)^m\lambda p_0p_1\cdots p_{m - 1} \cdot n(n - 1)\cdots(n - m) \,.
\end{multline}
Decreasing \eqref{den_limit} increases the MTTDL for a given $m$.
In the right hand side of \eqref{den_limit}, the value of $\left(\frac{\lambda}{\mu}\right)^{i}$ in the $i$\textsuperscript{th} term drops quickly with increasing $i$ 
for a small value of $\frac{\lambda}{\mu}$. Note that the $0$\textsuperscript{th} term disappears as $p_0$ is forced to 1 in any practical LDPC code.
It is easy to see that if $p_i$ in the $i$\textsuperscript{th} term is set to 1, then this term reduces to zero.
Since $\left(\frac{\lambda}{\mu}\right)^{i}$ is larger for a smaller value of $i$, forcing 
as many $p_{i}$'s for small $i$ as possible to 1 is crucial to minimize \eqref{den_limit} 
or, equivalently, maximize the MTTDL.
This property is the key to 
designing factor graphs that enhance reliability. 
Since the stopping number is the smallest number of erasures that cannot be corrected, 
it is clear that 
increasing the stopping number is equivalent to driving more $p_{i}$'s to 1.
Therefore, a large stopping number of the factor graph would mean an enhanced MTTDL. 
Theorem \ref{increasing_function} makes this relationship between the MTTDL and the stopping number more precise. 

\begin{theorem}\label{increasing_function}
	The MTTDL for LDPC codes is a monotonically increasing function of the stopping number $s^{\ast}$ of the given factor graph as $\frac{\lambda}{\mu} \rightarrow 0$ and assuming $\frac{\lambda}{\mu} < \frac{1}{n}$.
\end{theorem}

\begin{IEEEproof}
	See Appendix \ref{increasing}.
\end{IEEEproof}

Especially for the VN degree of 2, the stopping number $s^\ast$ is equal to $g/2$, where $g$ is the girth of the graph \cite{Stopping}. As a result, to increase reliability of the regular LDPC codes with $d_{v}=2$, the girth should be increased. This observation motivates LDPC code design by PEG, which is an effective search method for factor graphs with good girth properties.

	\begin{remark}
		As can be seen in \eqref{sdom} derived in the proof of Theorem \ref{increasing_function}, only the single probability $p_{s^\ast-1}$ really matters in computing the MTTDL. Empirical results also show that the simplified expression \eqref{sdom}, which is reproduced below, yields virtually identical MTTDL values as the full expression \eqref{MTTDL_LDPC}. 
		
	\begin{equation}\label{sdom0}
	\text{MTTDL} \rightarrow \frac{(m + 1)}{\left(\frac{\lambda}{\mu}\right)^{s^\ast - 1} \lambda(1 - p_{s^\ast - 1}) \prod_{i = 0}^{s^\ast - 1}(n - i) } \text{ as } \frac{\lambda}{\mu} \rightarrow 0\,.
	\end{equation}		
		Since the MTTDL is governed essentially by a single probability $p_{s^\ast-1}$, computing the MTTDL of an LDPC code now does not require estimating all $p_i$ probabilities through very extensive error pattern search.
	\end{remark}

\section{Quantitative Results}\label{simulation_results}

From the repair bandwidth analysis in Section III, it is shown that a regular CN degree minimizes the average repair bandwidth of LDPC codes. It is also shown that regular LDPC codes with $d_{v}=2$ can minimize repair bandwidth for a given code rate, provided degree-1 VNs are prohibited. In addition, from the MTTDL analysis in Section V, it is verified that LDPC codes should have a large stopping number which helps to improve reliability. With regards to
regular LDPC codes with $d_{v}=2$, the size of the girth plays the same role as the stopping number. 
We shall focus on PEG-LDPC codes in this section. PEG is a well-known algorithm which can construct factor graphs having a large girth \cite{PEG}. 
However, a concern that may arise for setting $d_v = 2$ is a potentially poor decoding capability 
in practical scenarios where multiple erasures may occasionally occur within a single codeword, since each VN is protected by only two sets of checks with $d_v=2$. We plot the data (codeword) loss probability of two $d_v = 2$ regular LDPC codes in Fig. \ref{WER} in environments where each symbol erasure occurs independently within each codeword. The results indicate that even a short LDPC code with $d_v = 2$ shows erasure correction behavior similar to 3-replication at low erasure probabilities.
Note that decoding capability improves when a larger LDPC code is adopted, showing data loss probability comparable to the (15,10) RS code.
In the case of irregular LDPC codes, even though the direct correlation between the girth and the stopping number is unknown, PEG is still a reasonable approach.
\begin{figure}[t]
	\centerline{\includegraphics[width=\figwidth]{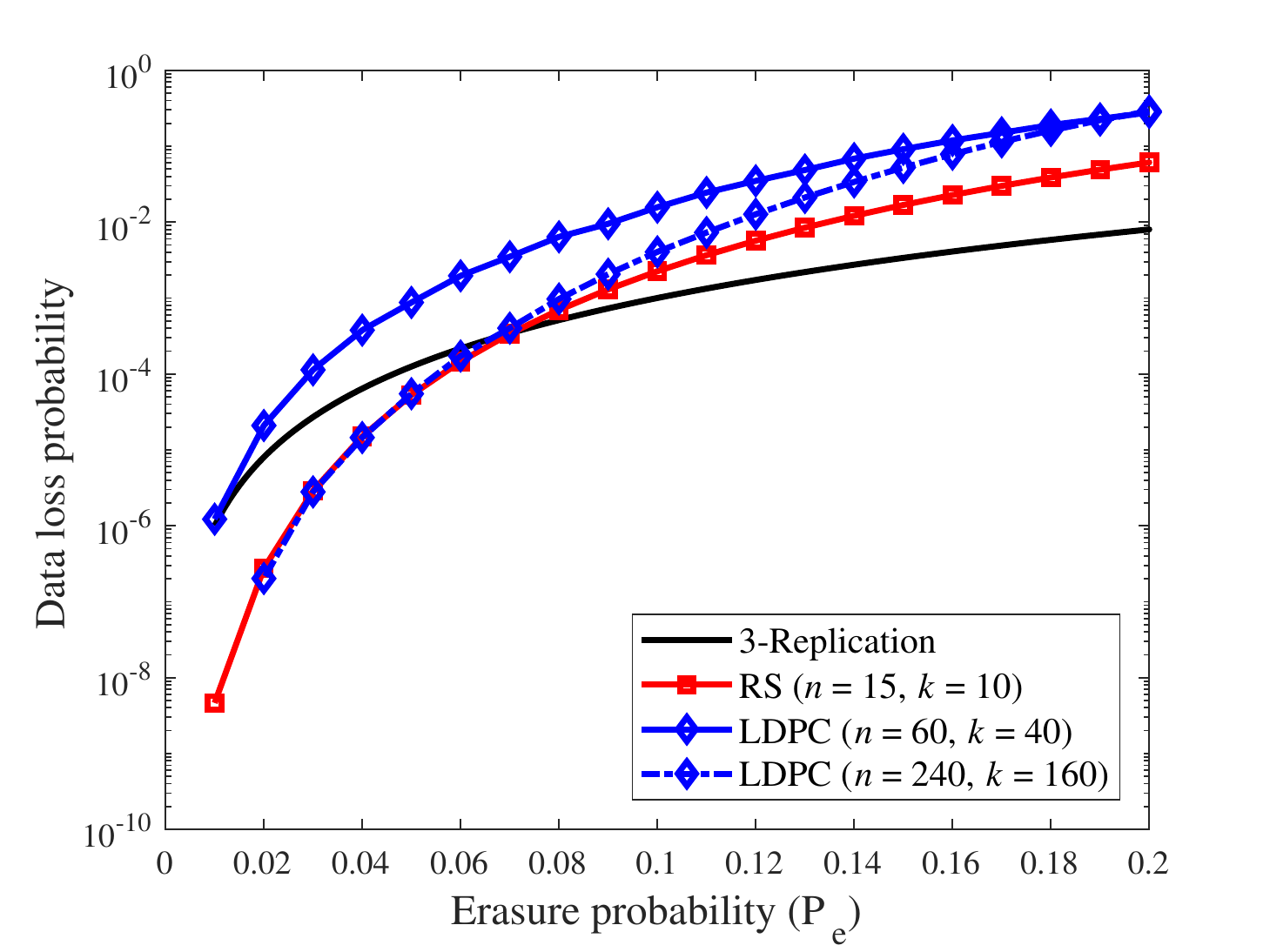}}
	\caption{Data loss probability of the (60, 40) and (240, 160) LDPC codes with $d_v = 2$ compared to the 3-replication and (15, 10) RS codes
	} \label{WER}
\end{figure}

Having ensured a good decoding capability, the metrics considered for comparison are storage overhead (code rate inverse), repair bandwidth and MTTDL. For the MTTDL simulation, the following normalized equation is used for fair comparison among codes having different lengths:
$$\text{MTTDL} = \frac{\text{MTTDL}_\text{stripe}}{C/nB}\,,$$
where $\text{MTTDL}_\text{stripe}$ is the MTTDL given in Section V for a stripe. Here, 
the MTTDL for a stripe is normalized by the number of stripes, $C/nB$, in storage system. 
The parameters used for MTTDL simulation are given in Table \ref{parameter}. These values are chosen consistent with the existing literature \cite{Azure, Xoring}. Note that for the repair rate, both the triggering time and the downloading time are included; the downloading time depends on the repair bandwidth (BW) overhead of the coding scheme. 
\begin{table}[t]
	\caption{Parameters used for MTTDL simulation
	}
	\label{parameter}
	\begin{center}
		\begin{tabular} {ccccccccccc}
			\hline\hline
			Parameter & Value & Description\\
			\hline
			$C$ & 40 PB & Total amounts of data\\
			$B$ & 256 MB & Block size \\
			$N_\text{disk}$ & 2000 & Number of disk nodes \\
			$S$ & 20 TB & Storage capacity of a disk \\
			$r_\text{node}$ & 1 Gbps & Network bandwidth on each node \\
			$1/\lambda$ & 1 year & MTTF (mean-time-to-failure) of a node \\
			$\mu$ & $\frac{1}{T_{t}+T_{r}}$ & Repair rate \\
			$T_{t}$ & 15 min & Detection and triggering time for repair \\
			$T_{r}$ & $\frac{S\cdot \text{BW}_\text{cost}}{r_\text{node}\cdot(N_\text{disk}-1)}$ & Downloading time of blocks \\
			$\text{BW}_\text{cost}$ & & Repair BW overhead of the given code \\
			$n$ & & Number of total coded blocks in a stripe \\
			$k$ & & Number of data blocks in a stripe \\
			$m$ & & Number of parity blocks in a stripe \\			
			\hline\hline
		\end{tabular}
	\end{center}
\end{table}
\begin{table}[t]
	\caption{Performance of QC-PEG LDPC codes with $d_{v}=2$, $R=2/3$.
	}
	\label{MTTDL}
	\begin{center}
		\begin{tabular} {ccccccccccc}
			\hline\hline
			Coding & Storage  & Repair BW  & MTTDL \\
			scheme & overhead & overhead			& (days)\\
			\hline
			3-replication & 3x & 1x & 1.20E+3 \\
			(15, 10) RS & 1.5x & 10x & 2.13E+10 \\
			(10, 6, 5) Xorbas LRC & 1.6x & 5x & 7.38E+7\\
			(15, 10, 6) Binary LRC & 1.5x & 6x & 3.00E+4\\
			(60, 40) LDPC & 1.5x & 5x &  1.40E+7\\
			(150, 100) LDPC & 1.5x & 5x &  1.42E+8\\
			(210, 140) LDPC & 1.5x & 5x & 2.91E+11 \\
			\hline\hline
		\end{tabular}
	\end{center}
\end{table}

For LDPC code simulations, using specific QC-PEG parity-check matrices, $p_i$'s are first obtained from decoding simulation and the MTTDL values are calculated from (\ref{MTTDL_LDPC}) or \eqref{sdom0}.\footnote{Note that the MTTDL value shown here for 3-replication is different from that in \cite{Xoring, Azure} due to the fact that the definition of the repair rate is different
(in \cite{Azure}, $\mu = 1/T_r$ for repair from a single failure and $\mu = 1/T_t$ from multiple failures, and in \cite{Xoring}, $\mu=r_\text{node}/B$).} Table \ref{MTTDL} shows performance of the QC-PEG LDPC codes with $d_v=2$ for $R=2/3$. Here, the (15, 10) RS code is chosen for comparison as well as simple replication and existing LRC methods. 

For a given storage overhead, LDPC codes in Table \ref{MTTDL} have a 5x repair bandwidth overhead, relative to replication, whereas the RS code has a 10x overhead. Thus, compared to the RS code, these LDPC codes require only one half of the repair bandwidth given the same storage overhead. 
Moreover, LDPC codes maintain the same repair bandwidth even as the code length is increased. This indicates that LDPC codes can get better MTTDLs than the (15, 10) RS code and the (10, 6, 5) Xorbas LRC \cite{Xoring} when longer codes are used. The table shows specifically that the (150, 100) and (210, 140) LDPC code has better performance in terms of both repair bandwidth and MTTDL compared to the (15, 10) RS code. Relative to the (10, 6, 5) Xorbas LRC, we observed that the (150, 100) and (210, 140) LDPC codes provide higher MTTDL. This is at the expense of a longer code length. In general, it is expected that the price of increasing the code length will be complexity. However, the complexity of encoding/decoding of LDPC codes in erasure channels is quite reasonable for the code lengths discussed here. The computational complexity issue of the LDPC code is discussed below.

\begin{remark}[Computational Complexity]
The computational complexity that need be considered in the context of distributed storage includes the encoding and decoding complexity. Note that LDPC encoding/decoding is based on simple XOR operations, while RS code and LRC require expensive Galois field operations. 
The encoding complexity of RS codes and LRCs both increases quadratically with respect to $n$; on the other hand, encoding of the LDPC code requires a linear (or near-linear) complexity. 
Decoding complexity is directly related to the computational burden required for reading data or repairing the failed block, which are the most frequent events in operating distributed storage. From this point of view, decoding complexity is also referred to as repairing complexity in distributed storage. 
The decoding/repairing traffic per one node of the LDPC code depends on the check node degree. Since $d_c$ is independent of $n$ as presented in Section III, overall decoding complexity of the LDPC code is only linear with $n$, whereas decoding the LRC and RS code requires complexity quadratic in $n$.
Specifically, the required numbers of additions and multiplications on average to decode/repair an LDPC code of rate $2/3$ that we employed are four and zero, respectively, regardless of the code length.
For decoding of the $(14, 10)$ RS code, nine additions and ten multiplications are required, which can increase tremendously with increasing code length.
The binary LRC \cite{BinaryLRC14,BinaryLRC16} is a modification of the Xorbas LRC to reduce computational complexity at the expense of repair bandwidth and MTTDL. For example, considering the failure of single nodes, decoding/repairing of a ($k$, $n - k$, $r_2$) $=$ (10, 5, 6) binary LRC (see Table \ref{MTTDL} for its repair bandwidth overhead and MTTDL) which is constructed based on a (10, 6, 5) Xorbas LRC requires five additions and zero multiplications. For the corresponding Xorbas LRC, four additions and 4.75 multiplications in binary extension field are needed on average. We thus observed that the LDPC code is competitive in terms of the decoding complexity as well thanks to its low-repair-bandwidth and the XOR-only feature. Note also that the difference in decoding complexity will increase further as the code length becomes longer.
\end{remark}
\begin{figure}[!t]
	\centerline{\includegraphics[width=\figwidth]{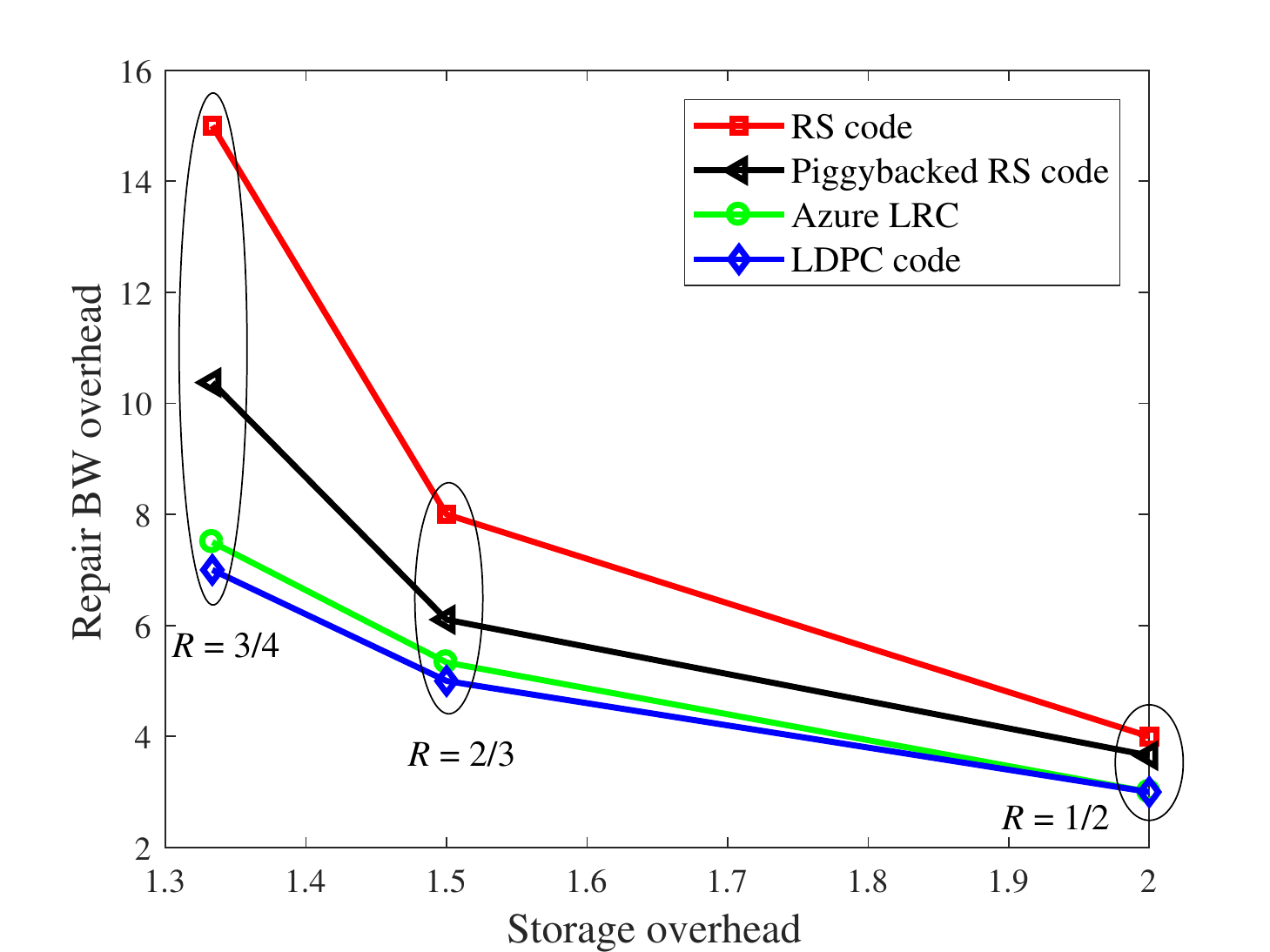}}
	\caption[Repair bandwidth comparison]{Tradeoffs between repair bandwidth overhead and storage overhead for different codes. Coding schemes having higher reliability than the (14, 10) RS code are considered.
	} \label{simulation1}
\end{figure}
\begin{table}[ht]
	\caption{Parameters of codes used in Fig. \ref{simulation1}. }
	\begin{center}
		\begin{tabular} {ccccccccccc}
			\hline\hline
			Scheme & $R$ = 3/4 & $R$ = 2/3 & $R$ = 1/2 \\
			\hline
			RS & (20, 15) & (12, 8) & (8, 4)\\
			Piggybacked RS & (20, 15) & (12, 8) & (8, 4) \\
			Azure LRC & (18, 3, 3) & (12, 3, 3) & (6, 3, 3) \\
			LDPC & (240, 180) & (120, 80) & (56, 28) \\
			\hline\hline
		\end{tabular}
	\end{center}
\end{table}

For rates 3/4, 2/3 and 1/2, various coding schemes are compared in Fig. \ref{simulation1}. Here we only consider codes that have higher MTTDLs than the (14, 10) RS code used in the Facebook cluster. The MTTDL of the (14, 10) RS codes is 1.61E+7. Note that our comparison with all other codes are done by averaging systematic and parity nodes. For the three storage overhead factors (code rate inverses), it is shown that LDPC codes have consistently better repair-bandwidth/storage-space tradeoffs compared to other codes. As the storage overhead is forced to decrease, LDPC codes enjoy a bigger performance gap relative to other codes with the exception of the LRC codes that perform similar to the LDPC codes.
\begin{figure}[!t]
	\centerline{\includegraphics[width=\figwidth]{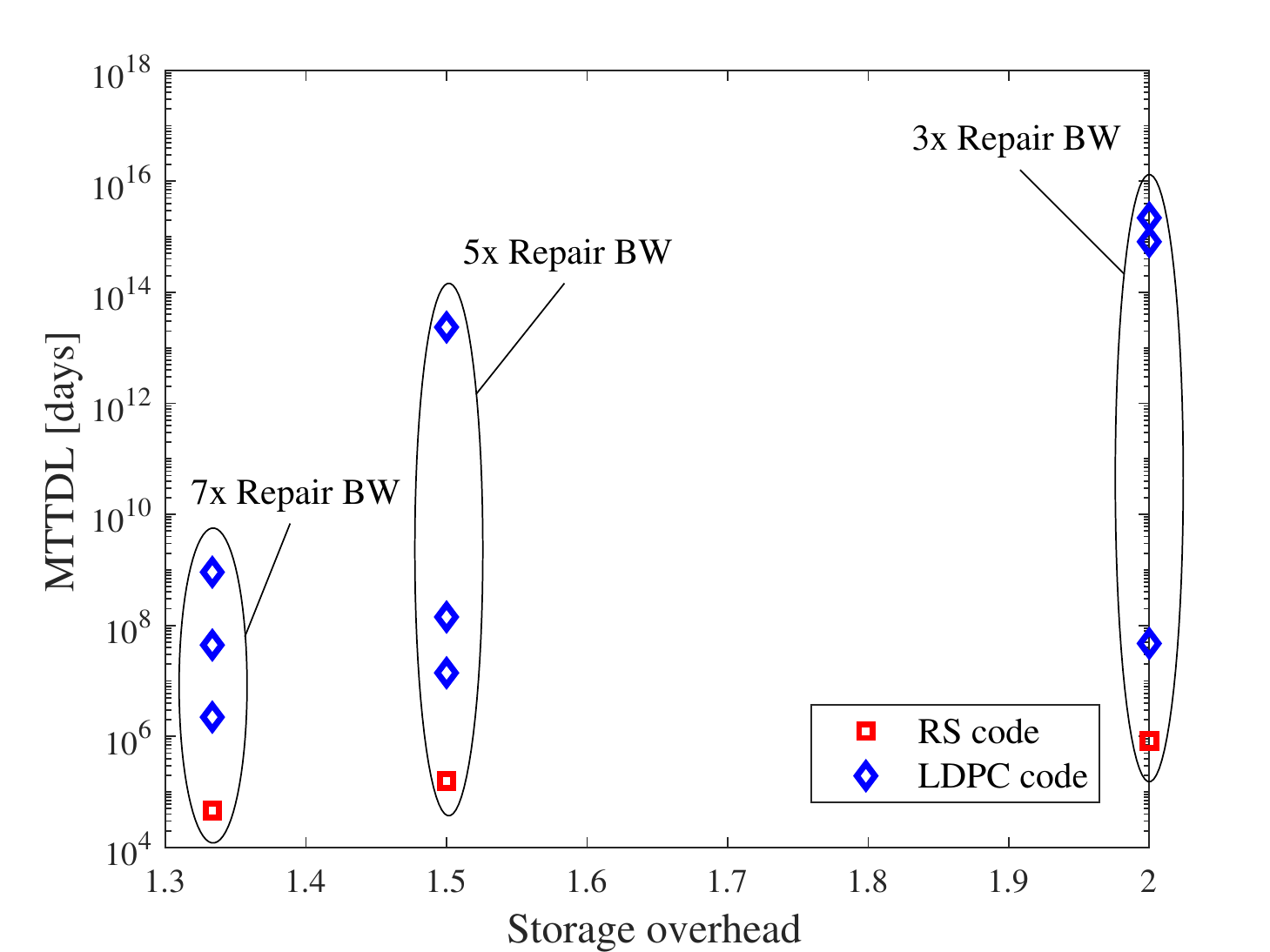}}
	\caption[MTTDL comparison]{MTTDL comparison of regular LDPC and RS codes under different storage overhead and repair bandwidth constraints 
	} \label{simulation2}
\end{figure}
\begin{table}[!t]
	\caption[Code details in Fig. \ref{simulation2}]{Parameters of codes used in Fig. \ref{simulation2}. LDPC1 represents the LDPC codes with the lowest MTTDLs.
	}
	\begin{center}
		\begin{tabular} {ccccccccccc}
			\hline\hline
			Scheme & $R$ = 3/4 & $R$ = 2/3 & $R$ = 1/2 \\
			\hline
			RS & (10, 7) & (8, 5) & (6, 3) \\
			LDPC1 & (80, 60) & (60, 40) & (44, 22) \\
			LDPC2 & (200, 150) & (150, 100) & (72, 36) \\
			LDPC3 & (320, 240) & (240, 160) & (100, 50) \\
			\hline\hline
		\end{tabular}\label{MTTDL_parameter_table}
	\end{center}
\end{table}

For given storage and repair bandwidth overheads, LDPC codes can achieve better MTTDL by increasing the code length, compared to the LRC and other codes. Fig. \ref{simulation2} shows such MTTDL comparison between the RS and LDPC codes, where for a given storage overhead,
the MTTDL advantage of the LDPC codes is evident. Since the MTTDL of the LRC is known to be similar to that of the RS codes \cite{Azure}, LDPC codes will have definite reliability advantages over the LRCs. 

\begin{figure}[!t]
	\centerline{\includegraphics[width=\figwidth]{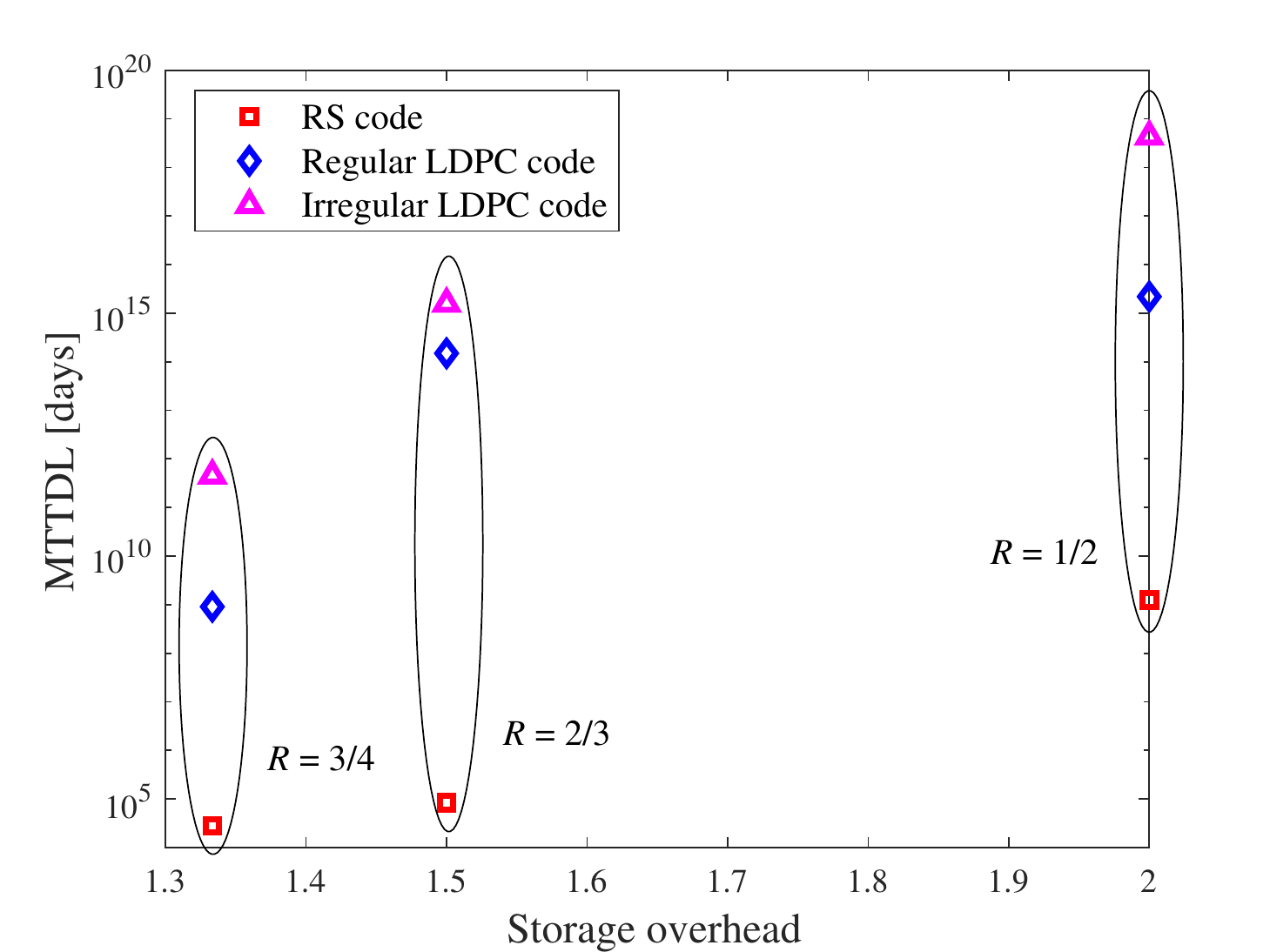}}
	\caption[MTTDL comparison]{MTTDL comparison of irregular LDPC, regular LDPC and RS codes under different storage overhead and repair bandwidth constraints 
	} \label{simulation3}
\end{figure}

MTTDLs of irregular LDPC codes that are designed to enhance the system reliability are shown in Figs. \ref{simulation3} and \ref{simulation4}.
Irregular LDPC codes are designed by the VN degree distributions given in Table \ref{degree_distributions}.
The code-lengths are set to be identical to those of LDPC3, and it is guaranteed that the global girth size is strictly larger than 4.

Fig. \ref{simulation3} shows the MTTDLs of the designed irregular LDPC codes with repair bandwidths increased by one relative to the regular LDPC codes also included in the figure. 
The MTTDLs of RS and regular LDPC codes are also shown for comparison.  
The parameters of the RS and LDPC codes are in Table \ref{MTTDL_parameter_table_ir}. The LDPC codes have the same code parameters as LDPC3 in Table \ref{MTTDL_parameter_table}, and the parameters of the RS codes are set to have the same repair bandwidth as the irregular LDPC codes being compared.
As can be seen, the designed irregular LDPC codes outperform RS and regular LDPC codes in terms of the MTTDL, at the cost of increased repair bandwidth (by 1).

\begin{figure}[!t]
\label{key}	\centerline{\includegraphics[width=\figwidth]{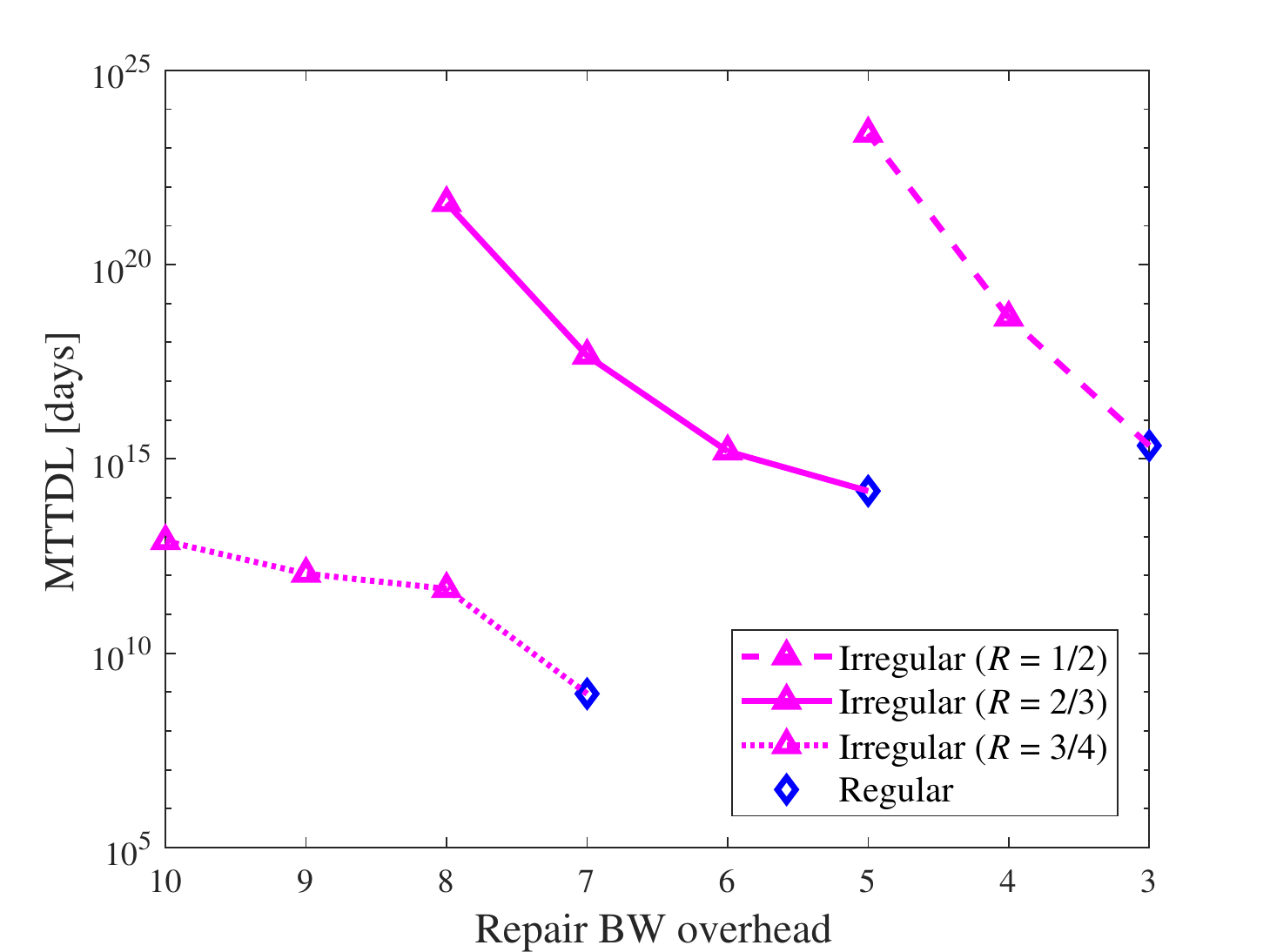}}
	\caption{Tradeoffs between MTTDL and repair bandwidth overhead of LDPC codes under different code rates.
	} \label{simulation4}
\end{figure}
\begin{table}[t]
	\caption[Code details in Figs. \ref{simulation3} and \ref{simulation4}]{Parameters of codes used in Figs. \ref{simulation3} and \ref{simulation4}. Both regular and irregular LDPC codes are based on the same code parameters.
	}
	\begin{center}
		\begin{tabular} {ccccccccccc}
			\hline\hline
			Scheme & $R$ = 3/4 & $R$ = 2/3 & $R$ = 1/2 \\
			\hline
			RS & (11, 8) & (9, 6) & (8, 4) \\
			LDPC & (320, 240) & (240, 160) & (100, 50) \\
			\hline\hline
		\end{tabular}\label{MTTDL_parameter_table_ir}
	\end{center}
\end{table}

Fig. \ref{simulation4} represents the behavior of MTTDLs versus repair bandwidth for LDPC codes.
The code parameters are the same as those used in Fig. \ref{simulation3}.
As mentioned above, regular LDPC codes with $d_v = 2$ have the minimum repair bandwidth for given code parameters.
We see that MTTDLs improve substantially when the repair bandwidths are allowed to grow from the minimum value.
 The tradeoff effect is more dramatic for smaller code rates. 

%
%

\section{Concluding Remarks}
\subsection{Conclusion}
For distributed storage applications, this paper shows that LDPC codes could be a highly viable option in terms of storage overhead, repair bandwidth and reliability tradeoffs. Unlike the RS code, 
the repair bandwidth of the LDPC codes does not increase with the code length.
As a result, the LDPC codes can be designed to enjoy both low repair bandwidth 
and high reliability compared to the RS code and its known variants. 
It has been specifically shown that for a given number of edges in the factor graph, CN-regular LDPC codes minimize the repair bandwidth.
In addition to the requirement of the CN-regularity, VN-regularity with $d_v = 2$ minimizes the repair bandwidth for a given code rate, barring all VNs with degree 1.
A code design that takes advantage of the improved reliability of LDPC codes has been given, yielding useful tradeoff options between MTTDL and repair bandwidth.
The MTTDL analysis for LDPC codes has also been provided that relates the code's stopping set size with its MTTDL.

\subsection{Future Work}
Interesting future work includes LDPC code design aiming at reduction of both repair-bandwidth and latency.
For the reliability analysis in this paper, we assumed that there occurred only one node failure at a time since it was the most frequent failure pattern.
Considering multiple erasures it can be shown that the repair bandwidth of LDPC codes is still one less than the CN degree.
However, the number of decoding iterations required to repair multiple erasures may differ from one specific code design to next.
Since the decoding latency of LDPC codes is proportional to the number of decoding iterations \cite{Smith10}, we need to design LDPC code degree distributions to minimize the number of decoding iterations.
It would be meaningful to study LDPC code structures that maximize the number of single-step recoverable nodes in combating latency.

The update complexity (defined as the maximum number of coded symbols updated for one changed symbol in the message \cite{Update10}) is an important measure especially in applications to highly dynamic distributed storage in which data updates are frequent. The study on the existence and construction of update-efficient codes is an active area of research (e.g., see \cite{Update10, Mazumdar14, Jule11, BalancedLRC16}). 
Investigating the relationships between update complexity and other performance metrics considered in this paper such as the MTTDL would be a good direction as well.


%

\appendices
\numberwithin{equation}{section}
\section{Proof of Lemma 2}\label{proof_MTTDL_parameter}

\begin{IEEEproof}
	As can be seen in Fig. \ref{non-MDS2}, the Markov model of LDPC codes with $m$ parity blocks consists of $m + 2$ states regarding the state $m + 2$ as a DL state.
	Let $\pi_i(t)$ denote the state probability of the state $i$ at time $t$ and $\mathcal{I}$ the set of all states. Then we have a constraint $\sum_{i\in \mathcal{I}} \pi_i(t) = 1$, since the process must be in one of the states at any given time $t \geq 0 $. For an arbitrary number $m$ of parity blocks, we build the sets of equations describing the Markov model which are followed by the MTTDL equation of LDPC codes with $m$ parity blocks.
	
	Assume that state 0 is the initial state of the Markov chain, so that
	\begin{equation}\nonumber
	\pi_{i}(0)=
	\begin{cases}
	1 & \text{if } i = 0 \\ 
	0 & \text{otherwise}\,.
	\end{cases}
	\end{equation}
	First we construct a set of differential equations from the Markov model in Fig. \ref{non-MDS2}.
	\begin{itemize}
		\item $i = 0$
		\begin{equation}\label{eq: DE0}
		\frac{\dd \pi_{i}(t)}{\dd t} = -n\lambda \pi_{i}(t) + \mu \pi_{i + 1}(t)\,,
		\end{equation}
		
		\item $1 \leq i < m$
		\begin{equation}\label{eq: DE1m}
		\frac{\dd \pi_{i}(t)}{\dd t} = -(n - i)\lambda \pi_{i}(t) - \mu \pi_{i}(t)
		+ \mu \pi_{i + 1}(t) + (n - i + 1)\lambda p_{i - 1} \pi_{i - 1}(t)\,,
		\end{equation}
		
		\item $i = m$
		\begin{equation}\label{eq: DEm}
		\frac{\dd \pi_{i}(t)}{\dd t} = -(n - i)\lambda \pi_{i}(t) - \mu \pi_{i}(t)
		+ (n - i + 1)\lambda p_{i - 1} \pi_{i - 1}(t)\,,
		\end{equation}
		
		\item $i = m + 1$ (Data loss state)
		\begin{multline}\label{eq: DEDL}
		\frac{\dd \pi_{i}(t)}{\dd t} =  (n - i + 1)\lambda \pi_{i - 1}(t)
		+(n - i + 2)\lambda(1-p_{i - 2})\pi_{i - 2}(t)\\ + \cdots
		+ (n-1)\lambda(1-p_1)\pi_{1}(t) + n\lambda(1-p_0)\pi_{0}(t)\,.
		\end{multline}
		
	\end{itemize}
	Taking the Laplace transforms of \crefrange{eq: DE0}{eq: DEDL} yields the following string of equations, where $\bar{\pi}_{i}(s)$ denotes the Laplace transform of $\pi_i(t)$.
	\begin{itemize}
		\item $i = 0$
		\begin{equation}\label{eq: Laplace0}
		s\bar{\pi}_{i}(s) = -n\lambda \bar{\pi}_{i}(s) + \mu \bar{\pi}_{i + 1}(s) + 1\,,
		\end{equation}
		
		\item $1 \leq i < m$
		\begin{equation}\label{eq: Laplace1m}
		s\bar{\pi}_{i}(s) = -(n - i)\lambda \bar{\pi}_{i}(s) - \mu \bar{\pi}_{i}(s) 
		+ \mu \bar{\pi}_{i + 1}(s) + (n - i + 1) \lambda p_{i - 1} \bar{\pi}_{i - 1}(s)\,,
		\end{equation}
		
		\item $i = m$
		\begin{equation}\label{eq: Laplacem}
		s\bar{\pi}_{i}(s) = -(n - i)\lambda \bar{\pi}_{i}(s) - \mu \bar{\pi}_{i}(s)
		+(n - i + 1) \lambda p_{i - 1} \bar{\pi}_{i - 1}(s)\,,
		\end{equation}
		
		\item $i = m + 1$ (Data loss state)
		\begin{multline}\label{eq: LaplaceDL}
		s\bar{\pi}_{i}(s) = (n - i + 1)\lambda \bar{\pi}_{i - 1}(s)
		+ (n- i + 2)\lambda(1-p_{i - 2})\bar{\pi}_{i - 2}(s)\\ + \cdots
		+ (n-1)\lambda(1-p_1)\bar{\pi}_{1}(s) + n\lambda(1-p_0)\bar{\pi}_{0}(s)\,.
		\end{multline}
		
	\end{itemize}
	Solving \crefrange{eq: Laplace0}{eq: LaplaceDL} for $s\bar{\pi}_{m + 1}(s)$,
	$s\bar{\pi}_{m+1}(s)$ is presented as follows:
	\begin{multline}\label{eq: spi}
	s\bar{\pi}_{m+1}(s) =\frac{ n\lambda(1-p_{0})\cdot G_1(s)}{G_0(s)}
	+\frac{\sum_{j=1}^{m-1}[\{\prod_{i=0}^{j}(n-i)\cdot\lambda^{j+1}\} \cdot \{\prod_{i=0}^{j-1}p_{i}\cdot(1-p_{j})\cdot G_{j+1}(s)\}]}{ G_0(s) }\\
		+ \frac{ \{\prod_{i=0}^{m}(n-i)\cdot\lambda^{m+1}\} \cdot \prod_{i=0}^{m-1}p_{i} }{ G_0(s) }\,,
	\end{multline}
	where $G_i(s)$ for $0 \leq i \leq m$ is recursively defined by
	\begin{equation}\label{eq: Gs}
	\begin{split}
	G_m(s) =& s + (n - m)\lambda + \mu\,, \\
	G_{m - 1}(s) =& \{s + (n - m + 1)\lambda + \mu\}G_m(s) - \mu (n - m + 1)\lambda p_{m - 1}\,, \\
	G_{1 \leq i < m - 1}(s) = & \{s + (n - i)\lambda + \mu\}G_{i + 1}(s)
	- \mu (n - i)\lambda p_{i}G_{i + 2}(s)\,, \\
	G_{0}(s) =& (s + n\lambda)G_1(s) - \mu n \lambda p_0 G_2(s)\,.
	\end{split}
	\end{equation} 
	Equation \eqref{eq: Gs} is determined by the following set of equations:
	\begin{equation*}
	\begin{split}
	\bar{\pi}_{m}(s) = & \frac{(n-m+1)\lambda p_{m-1}}{s+(n-m)\lambda+\mu} \bar{\pi}_{m-1}(s) 
	\coloneqq  \frac{(n-m+1)\lambda p_{m-1}}{G_m(s)}\bar{\pi}_{m-1}(s)\,, \qquad\\
	\bar{\pi}_{m-1}(s) = & \frac{(n-m+2)\lambda p_{m-2}}{s+(n-m+1)\lambda+\mu - \mu \cdot \frac{(n-m+1)\lambda p_{m-1}}{G_m(s)} } \bar{\pi}_{m-2}(s)\\
	\coloneqq & \frac{(n-m+2)\lambda p_{m-2}\cdot G_m(s)} {G_{m-1}(s)}\bar{\pi}_{m-2}(s)\,,
	\end{split}
	\end{equation*}
	\begin{equation*}
	\begin{split}
	\bar{\pi}_{m-2}(s) = & \frac{(n-m+3)\lambda p_{m-3}}{s+(n-m+2)\lambda+\mu - \mu \cdot \frac{(n-m+2)\lambda p_{m-2}G_m(s)}{G_{m-1}(s)} } \bar{\pi}_{m-3}(s) \\
	\coloneqq & \frac{(n-m+3)\lambda p_{m-3}\cdot G_{m-1}(s)} {G_{m-2}(s)}\bar{\pi}_{m-3}(s)\,, \\
	\vdots\\
	\bar{\pi}_{1}(s) = & \frac{n\lambda p_0}{s+(n-1)\lambda+\mu - \mu \cdot \frac{(n-1)\lambda p_{1}G_3(s)}{G_{2}(s)} } \bar{\pi}_{0}(s)
	\coloneqq \frac{n\lambda p_{0}\cdot G_{2}(s)} {G_{1}(s)}\bar{\pi}_{0}(s)\,, \\
	\bar{\pi}_{0}(s) = & \frac{1}{s+n\lambda-\mu \cdot \frac{n\lambda p_0 G_2(s)}{G_1(s)}} \coloneqq \frac{G_{1}(s)} {G_{0}(s)}\,.
	\end{split}
	\end{equation*}		
	From the moment generating property of Laplace transforms, the MTTDL is given by
	\begin{equation}\label{eq: MTTDL}
	\text{MTTDL} = -\frac{\dd}{\dd s}\Big(s\bar{\pi}_{m + 1}(s)\Big)\Big\vert_{s=0}\,.
	\end{equation}
	Combining \eqref{eq: spi} and \eqref{eq: MTTDL} we arrive at the final result:
	\begin{align}
	\nonumber & \text{MTTDL} \\
	& = \frac{G_0(0)\cdot G_0'(0)}{G_0^2(0)} - \frac{  n\lambda(1-p_{0})\cdot G'_1(0)\cdot G_0(0)}{G_0^2(0)}\nonumber\\
	& \pushright{ - \frac{\Big[\sum_{j=1}^{m-1}[\{\prod_{i=0}^{j}(n-i)\cdot\lambda^{j+1}\} \cdot \{\prod_{i=0}^{j-1}p_{i}\cdot(1-p_{j})\cdot G'_{j+1}(0)\}]   \Big] \cdot G_0(0) } {G_0^2(0)} \qquad\quad } \label{finalValue} \\
	& \leq \Bigg\{ G_0'(0) -  n\lambda(1-p_{0})\cdot G'_1(0) - \sum_{j=1}^{m-1}\bigg[\Big\{\prod_{i=0}^{j}(n-i)\cdot\lambda^{j+1}\Big\}
	\cdot \Big\{\prod_{i=0}^{j-1}p_{i}\cdot(1-p_{j})\cdot G'_{j+1}(0)\Big\}\bigg] \Bigg\} \nonumber\\
	& \pushright{ \Bigg/ \Bigg\{ n\lambda(1-p_{0})\cdot \mu^{m}
	+\sum_{j=1}^{m-1}\bigg[\Big\{\prod_{i=0}^{j}(n-i)\cdot\lambda^{j+1}\Big\} \cdot \Big\{\prod_{i=0}^{j-1}p_{i}\cdot(1-p_{j})\cdot \mu^{m-j}\Big\}\bigg] } \qquad \qquad \nonumber\\
	&\pushright{ + \Big\{\prod_{i=0}^{m}(n-i)\cdot\lambda^{m+1} \Big\} \cdot \prod_{i=0}^{m-1}p_{i} \Bigg\} } \qquad \quad \label{simeqMu} \\
	& \rightarrow (m + 1)\mu^m \Bigg/ \Bigg\{ n\lambda(1-p_{0})\cdot \mu^{m}
		+\sum_{j=1}^{m-1}\bigg[\Big\{\prod_{i=0}^{j}(n-i)\cdot\lambda^{j+1}\Big\} \cdot \Big\{\prod_{i=0}^{j-1}p_{i}\cdot(1-p_{j})\cdot \mu^{m-j}\Big\}\bigg] \nonumber\\
	&\pushright{ + \Big\{\prod_{i=0}^{m}(n-i)\cdot\lambda^{m+1} \Big\} \cdot \prod_{i=0}^{m-1}p_{i} \Bigg\} } \qquad \quad \nonumber\\
	& \pushright{ \text{ as } \frac{\lambda}{\mu} \rightarrow 0\,, } \qquad \quad \label{largeMu}
	\end{align}
	where (\ref{finalValue}) follows from the final-value theorem of Laplace transforms; as $s \to 0$, \eqref{eq: spi} implies \eqref{FVTret} considering the final-value theorem in \eqref{FVTlim}.
	\begin{multline} \label{FVTret}
	G_0(s) = n\lambda(1-p_{0})\cdot G_1(s)
	+\sum_{j=1}^{m-1}\bigg[\Big\{\prod_{i=0}^{j}(n-i)\cdot\lambda^{j+1}\Big\} \cdot \Big\{\prod_{i=0}^{j-1}p_{i}\cdot(1-p_{j})\cdot G_{j+1}(s)\Big\}\bigg] \\ 
	+ \Big\{\prod_{i=0}^{m}(n-i)\cdot\lambda^{m+1}\Big\} \cdot \prod_{i=0}^{m-1}p_{i}\,,
	\end{multline}
	\begin{equation} \label{FVTlim}
	\lim\limits_{t \to \infty} \pi_{m+1}(t) = \lim\limits_{s \to 0}s\bar{\pi}_{m+1}(s)=1\,.
	\end{equation}
	Equation (\ref{simeqMu}) is due to \eqref{FVTret} and from the fact that $G_1(0) \geq \mu^m$, $G_2(0) \geq \mu^{m-1}$, $\cdots$, $G_{m - 1}(0) \geq \mu^2$, and $G_m(0) \geq \mu$ in \eqref{FVTret}.
	Equation (\ref{largeMu}) is because the numerator of \eqref{simeqMu} approaches $(m + 1)\mu^m$ as $\frac{\lambda}{\mu} \rightarrow 0$. 
\end{IEEEproof}

\section{Proof of Theorem 2}\label{increasing}

\begin{IEEEproof}
Consider the right hand side of \eqref{den_limit}. Recall that the stopping number $s^\ast$ is the smallest number of erasures that cannot be corrected 
and that $p_i$ is the conditional probability that an additional erasure will be tolerated given $i$ erasures. This leads to the property that $p_i = 1$ for any $i < s^\ast - 1$, and the first probability that is not equal to 1 as $i$ increases from 0 is $p_{s^\ast - 1}$. Then, the first non-zero term in the right hand side of \eqref{den_limit} is  
\begin{equation}
\left(\frac{\lambda}{\mu}\right)^{s^\ast - 1}\lambda (1-p_{s^\ast - 1}) \cdot n(n - 1)\cdots(n - s^\ast +1) \,.
\end{equation}
We now show that this term dominates as $\frac{\lambda}{\mu} \rightarrow 0$. In fact, the ratio of the next term and this term is given by
\begin{equation}\label{ratio}
\frac{\lambda}{\mu} \cdot \frac{ (1 - p_{s^\ast})p_{s^\ast - 1}(n - s^\ast) }{1 - p_{s^\ast - 1}}\,,
\end{equation}
which approaches zero for any finite $n$ as  $\frac{\lambda}{\mu} \rightarrow 0$. Using the same argument the similar ratio of any two successive terms reduces to zero in the limit. This means that the MTTDL of an LDPC code simplifies to     
\begin{equation}\label{sdom}
\frac{(m + 1)}{\left(\frac{\lambda}{\mu}\right)^{s^\ast - 1} \lambda(1 - p_{s^\ast - 1}) \prod_{i = 0}^{s^\ast - 1}(n - i) }
\end{equation}
for $\frac{\lambda}{\mu} \rightarrow 0$. Now, for any reasonably large $n$, we have $\prod_{i = 0}^{s^\ast - 1}(n - i) \approx  n^{s^\ast}$. The MTTDL in the limit of $\frac{\lambda}{\mu} \rightarrow 0$ can now be rewritten as 
\begin{equation}\label{sdom2}
\frac{(m + 1)}{\left(\frac{\lambda}{\mu}\cdot n\right)^{s^\ast} \mu(1 - p_{s^\ast - 1}) }\,,
\end{equation}
which is a monotonically and very rapidly increasing function of $s^\ast$, as long as $\frac{\lambda}{\mu} < 1/n$. This completes the proof.

%
\end{IEEEproof}

\ifCLASSOPTIONcaptionsoff
  \newpage
\fi



\bibliographystyle{IEEEtranTCOM}
\bibliography{repairBandwidthLDPC17}
\end{document}